\documentclass[11pt]{llncs}

\usepackage[margin=3cm]{geometry}
\usepackage{amssymb}
\usepackage{comment}
\usepackage{graphicx}

\usepackage{tikz}  
\usetikzlibrary{arrows}
\pgfdeclarelayer{bg}    
\pgfsetlayers{bg,main} 

\usepackage{algorithm}
\usepackage[noend]{algorithmic}

\usepackage{amsmath, amssymb}
\usepackage{mathtools}
\usepackage{verbatim}
\usepackage{color}
\usepackage{xcolor}

\usepackage{multicol}

\usepackage{xspace}
\usepackage{cleveref}
\usepackage{thmtools, thm-restate}
\usepackage[numbers,sort&compress]{natbib}

 \let\mathscr\relax
\usepackage[scr]{rsfso}

\newcommand{\suchthat}{\;\ifnum\currentgrouptype=16 \middle\fi|\;}

\def\epsilon{\varepsilon}

\renewcommand{\bar}{\overline}

\newcommand{\myheader}[1]{\smallskip\noindent\textbf{#1.}~~}

\begin{document}

\title{Fairness and Efficiency in DAG-based Cryptocurrencies} 

\author{Georgios Birmpas \inst{1} \and Elias Koutsoupias \inst{1} \and Philip Lazos \inst{2} \and \\ Francisco J. Marmolejo-Coss\'io \inst{1}}

\institute{
University of Oxford\\
\and
Sapienza University of Rome}

\maketitle

\begin{abstract}

 Bitcoin is a decentralised digital currency that serves as an alternative to existing transaction systems based on an external central authority for security. Although Bitcoin has many desirable properties, one of its fundamental shortcomings is its inability to process transactions at high rates. To address this challenge, many subsequent protocols either modify the rules of block acceptance (longest chain rule) and reward, or alter the graphical structure of the public ledger from a tree to a directed acyclic graph (DAG). 
 
 Motivated by these approaches, we introduce a new general framework that captures ledger growth for a large class of DAG-based implementations. With this in hand, and by assuming \emph{honest} miner behaviour, we (experimentally) explore how different DAG-based protocols perform in terms of \emph{fairness}, i.e., if the block reward of a miner is proportional to their hash power, as well as \emph{efficiency}, i.e. what proportion of user transactions a ledger deems valid after a certain length of time. 
 
 Our results demonstrate fundamental structural limits on how well DAG-based ledger protocols cope with a high transaction load. More specifically, we show that even in a scenario where every miner on the system is honest in terms of when they publish blocks, what they point to, and what transactions each block contains, fairness and efficiency of the ledger can break down at specific hash rates if miners have differing levels of connectivity to the P2P network sustaining the protocol.

\end{abstract}

\section{Introduction}

Bitcoin and many other decentralised digital currencies maintain a public ledger via distributed consensus algorithms implemented using blockchain data structures. End users of the currency post transactions to the P2P network sustaining the protocol and said transactions are bundled into blocks by {\em miners}: agents tasked with the upkeep of the ledger. With respect to Bitcoin, the prescribed \emph{longest chain rule} dictates that miners must bundle pending transactions into a block that also includes a single hash pointer to the end of the longest chain seen by the miner in their local view of the ledger. Furthermore, in order for a block to be valid, its hash must lie below a dynamically adjusted threshold. Hence, miners must expend computational resources to find valid blocks. Due to this \emph{Proof-of-Work} structure, if all miners follow the protocol, the number of blocks they contribute to the blockchain is proportional to the computational resources they dedicate to the protocol, i.e. their hash power. In addition, miners are incentivised to follow the protocol via judicious incentive engineering through block rewards. This latter point also implies that miners earn block reward proportional to their hash power, thus making Bitcoin a {\em fair} protocol for miners.

As mentioned before, Bitcoin dynamically adjusts its target hash for valid blocks so that the totality of all miners active in the protocol find a block every ten minutes on average. This feature of the protocol makes consensus more robust, as this time-scale is much larger than the time it takes for a block to propagate on the P2P network supporting Bitcoin. However, since the size of blocks is limited, Bitcoin inherently suffers from a scalability problem. Thus in spite of Bitcoin being strategy-proof and fair, it suffers in its {\em efficiency}: which we define as the expected ratio of the number of valid transactions in the ledger to the number of all transactions posted in the P2P network. On the other hand, simply decreasing confirmation times and demanding higher transaction throughput by either increasing the overall block creation rate or block size can also affect these very properties of the protocol. For instance, delays in the P2P network may cause miners to have different views of the ledger, which can in turn directly make achieving a consensus more difficult, or lead miners to be strategic when they would have otherwise acted honestly. Ultimately, it seems that Bitcoin fundamentally strikes a delicate balance between being strategy-proof and fair at the cost of efficiency.

There have been many attempts to cope with Bitcoin's inherent throughput limitations, with \cite{LewenbergSZ15,GHOST,PHANTOM,SPECTRE} being some notable examples. All of these papers focus on how security can be maintained when the throughput is increased and follow the common direction of either modifying Bitcoin's longest chain rule or implementing a different graphical structure underlying the ledger.

In GHOST \cite{GHOST} an alternative consensus rule is proposed to the longest chain of Bitcoin, focusing on creating a new protocol that maintains security guarantees even when faced with high transaction loads. In this setting, GHOST takes into account the fact that forks are more likely to be produced when the underlying ledger still takes the form of a tree, as with Bitcoin. More specifically, when deciding what a newly mined block should point to, GHOST no longer myopically points to the head of the longest chain, but rather starts from the genesis block and at each fork, chooses the branch of the fork that leads to the heaviest subtree in the ledger until reaching a leaf to point to. In this way, blocks that are off the main chain can still contribute to the final consensus, which arguably maintains a degree of robustness to strategic mining while coping with high throughput better than Bitcoin.

In \cite{SPECTRE, PHANTOM} protocols SPECTRE and PHANTOM are proposed, with ledger structures in the form of directed acyclic graphs (DAG). The protocols in both of these implementations suggest that every newly created block has to point to every available (visible) leaf in the ledger. In that way every created block will eventually become part of the consensus, and the security of the system remains affected by forks that will be produced due to high throughput, since they will in turn be part of the ledger. The obvious advantage is that the system becomes immune to attacks that focus on increasing the block rewards of a miner. On the other hand, ordering the transactions and preventing other types of strategic behaviour 
becomes more complicated.

Motivated by these ideas, we design a new theoretical framework that captures a large family of DAG-based ledger implementations (including those mentioned in previous paragraphs). We achieve this by introducing a parametric model which lets us adjust the number of blocks each newly created block can point to, the block attributes a miner takes into account when choosing what blocks to point to, and the number of transactions a block can store. Most importantly, we describe a theoretical framework for ledger growth in these DAG-based models, along with a novel simplification for extrapolating valid transactions from a ledger under the assumption that all miners are honest. With this in hand, we are able to answer how our family of DAG-based ledgers copes with the high transaction loads they are intended to tackle. Indeed, our results are structural in nature, for we show how fairness and efficiency suffer from high transaction rates in spite of all agents behaving honestly in a given DAG-based ledger.

\myheader{Our Results} Our simulations allow us to show specific transaction load regimes that break down efficiency for all protocols in our large class of DAG-based ledgers. Furthermore, we show that in almost all transaction load regimes, fairness is not only lost for our DAG-based ledgers, but exhibits a complicated relationship with respect to agent connectivity to the underlying P2P network. It seems that increasing the outdegree of the blockchain's nodes from 1 to 2 helps substantially, while further improvements in efficiency are marginal and mostly apparent if the throughput is pushed very high, where fairness is inherently limited and the number of pointers required would fill up most of the block.

\subsection{Related Work}
Bitcoin was introduced in Nakamoto's landmark white paper \cite{nakamoto2008bitcoin} as a decentralised digital currency. Since its inception many researchers have studied several aspects of the protocol, i.e. its security and susceptibility to different types of attacks \cite{BabaioffDOZ12,EyalS14,Eyal15,GarayKL15,liu2018strategy,MBSK19,NKMS16}, how it behaves under a game-theoretic perspective \cite{carlsten2016instability,kiayias2016blockchain,LazosKOS19} and how its scalability and inherent transaction throughput issues can be improved. Since the latter is the most related to our work, we give a more detailed exposition in the paragraphs that follow. Before we proceed, we also want to refer the reader to \cite{BonneauMCNKF15,TschorschS16} for some extensive surveys which provide a good view of the research and challenges that exist in the area.

Sompolinksy and Zohar \cite{GHOST}, study the scalability of Bitcoin, analysing at the same time the security of the protocol when the delays in the network are not negligible. More specifically, they build on the results of Decker and Wattenhofer  \cite{DeckerW13} and explore the limits of the amount of transactions that can be processed under the protocol, while also studying how transaction waiting times can be optimised when there is also a security constraint security. In the same work, the Greedy Heaviest-Observed Sub-Tree chain (GHOST) is also presented as a modified version of the Bitcoin protocol selection rule, and as a way of obtaining a more scalable system with higher security guarantees. It is interesting to mention that many existing cryptocurrencies currently use variations of the GHOST rule, with  Ethereum \cite{Ethereum} and Bitcoin-NG \cite{EyalGSR16} being some notable examples. 
The authors argue that under this rule, the possible delays of the network cannot affect the security of the protocol even if the designer allows high creation rates of large-sized blocks and thus a high transaction throughput. 

Subsequently, Kiayias and Panagiotakos \cite{KiayiasP17} further study the GHOST protocol and provide a general framework for security analysis that focuses on protocols with a tree structure. They expand upon the analysis of \cite{GHOST} and follow a direction similar to the one presented in the work of Garay et al. \cite{GarayKL15}, which only studies chain structures and cannot be directly implemented in the setting of GHOST. We would like to point out that in \cite{GarayKL15} Garay et al. also provide an extended analysis of their framework for the partially synchronous model under the existence of bounded delays in the underlying P2P network of the protocol.

Lewenberg et al. \cite{LewenbergSZ15} propose the structure of a DAG, instead of a tree, as a way of dealing with high block creation rates and blocks of large size. Building on this idea, the same authors in \cite{SPECTRE} present SPECTRE, a new PoW-based protocol that follows the DAG-structure, and is both, scalable and secure. More specifically, they argue that SPECTRE can cope with high throughput of transactions while also maintaining fast confirmation times. They also analyse its security by studying several types of attacks. Part of the contribution of the paper is also introducing a way to (partially) order created blocks 
via a voting rule among existing blocks, which also contributes to the security of the protocol. SPECTRE has drawn the attention of many researchers after its introduction and we refer the reader to \cite{GiladHMVZ17,KarlssonJWAMRW18,PassS17,PHANTOM} for some indicative related works. 

Finally, DAG-based ledgers are also used in the white paper for IOTA, \cite{IOTA}, describing a cryptocurrency specifically designed for the Internet-of-Things industry. Although this setting is different (i.e. absence of PoW and block creation rewards etc.) and the focus is on how one can resolve conflicts between transactions and produce a system of high-active participation, this work is still highly related to ours, since IOTA is studied under an asynchronous network.

\section{DAG-based Ledgers}\label{sec:DBL}

In this section we will describe a family of decentralised consensus algorithms for public ledgers that generalise Bitcoin and SPECTRE. In what follows, we assume that there are $n$ strategic miners $m_1,...,m_n$ with hash powers $h_1,...,h_n$ respectively. When a given block is found globally by the protocol, $h_i$ represents the probability that this block belongs to $m_i$. We will be studying DAG-based ledger implementations. Formally, these ledgers are such that blocks and their pointers induce a directed acyclic graph with blocks as nodes and pointers as edges. The maximum out-degree of a block, $k$ is specified by the protocol and is in the range $1 \leq k \leq \infty$. Thus it is straightforward to see that Bitcoin for example, is a DAG-based ledger where the DAG is in fact a tree (with $k=1$). Finally, since blocks have bounded size, we define $1\leq \eta < \infty $ to be the maximum number of transactions a block can store.

As mentioned in the introduction, we are primarily interested in studying issues of fairness and ledger efficiency in DAG-based protocols catered to a high throughput regime. We recall that a protocol is fair if a miner can expect to see a block reward proportional to their hash power, and that a protocol efficiency is the ratio of all valid transactions to all transactions broadcast over the P2P network. In this setting, and under the assumption of a discrete time horizon, transactions and blocks that are propagated by users in the P2P network may take multiple \emph{turns} (the time it takes for the entire system to find a block) before they are seen by certain agents within the system. For this reason, miners only see a portion of the entire block DAG produced by a decentralised protocol as well as a portion of all transactions propagated by all end users of the ledger. 

In actuality, transactions that are posted to the P2P network of digital currencies directly depend on other transactions. For this reason, we also model the set of transactions that end users generate as a DAG. Furthermore, the structure of the transaction DAG itself has important implications for how transactions are packed in blocks for any DAG-based ledger. For example, if the transaction DAG is a path, and we are considering SPECTRE as our DAG-based protocol, it is easy to see that transactions will only be packed proportional to the deepest node of the block DAG, which in the high throughput regime can grow at a much slower rate than that at which transactions are generated. At the other extreme, if the transaction DAG only consists of isolated nodes, then any block can contain any transaction, and the efficiency of SPECTRE is thus constrained by what transactions miners see rather than the structure of the block DAG. 

Ultimately, in addition to having computational power, a miner also has informational power, which encapsulates how connected they are to the P2P network and consequently, how much of each of the aforementioned DAGs they see at a given time. We model the informational parameter of an arbitrary miner $m_i$ as a parameter $q_i \in [0,1]$. As $q_i$ approaches 1, $m_i$ is likely to see the entirety of both DAGs, whereas as $q_i$ approaches 0, $m_i$ is likely to only see the blocks he mines and transactions he creates. 

\subsection{Ledger Growth Preliminaries}\label{sec:lgp}

 We begin by setting some preliminary notation about graphs that it will be used in several parts as we define the model. Let $\mathcal{G}$ be the set of all finite directed graphs. For $G \in \mathcal{G}$, $V(G)$ and $E(G) \subseteq V(G)^2$ are the set of vertices and directed edges of $G$ respectively. Furthermore, for a tuple $x = (x_i)_{i=1}^n$, we let $\pi_i(x) = x_i$ be the projection onto the $i$-th coordinate. Finally, we define the \emph{closure} of a subset $X\subseteq V(G)$ of vertices, which will be needed in order to describe how a miner perceives the current state of the network.

 \begin{definition}[Closure] \label{def:cls}
Suppose that $G \in \mathcal{G}$, and let $X \subseteq V(G)$ be a subset of vertices. We denote the closure of $X$ in $G$ by $\Gamma (X \mid G)$ and define it as the subgraph induced by all vertices reachable from $X$ via directed paths.
\end{definition}

 We now proceed by formally describing and exploring the stochastic growth of a DAG-based ledger given $m_1,...,m_n$ strategic miners in a step-by-step fashion. As we already mentioned, we assume that the ledger grows over a finite discrete time horizon: $t = 1,...,T$. Each turn will consist of four phases: a \emph{block revelation phase}, in which nature picks a miner to initialise a block, an \emph{information update phase}, where miners update their views of the block and transaction graphs, an \emph{action phase},  in which miners employ their strategies depending on their local information, and a \emph{natural transaction generation phase}, in which non-miners stochastically publish transactions to the P2P network.

At the end of the action phase of each turn  $t$, we maintain a global block-DAG and transaction directed graph, denoted by $G^{glob}_t$ and $T^{glob}_t$ respectively. We say that the vertices of $G^{glob}_t$ are blocks and we have that $G^{glob}_t$ contains every block (public or private) that has been created up to turn $t$. Similarly, for $T^{glob}_t$ we have that it consists of every transaction present in the network up to point $t$.
We denote $V(G^{glob}_t) = \{B_1,...,B_t\}$, where the $i$-th block was created at the $i$-th turn and $V(T^{glob}_t) = T^*_t \cup \bar{T_t}$, where $T^*_t = \{tx^*_1,...,tx^*_t\}$ (enumerated) represents the set of the respective block rewards and $\bar{T_t}$ the set of the transactions.

Each block $B_t$, has out-degree of at most $k$ and carries at most $\eta+1$ transactions denoted by $Tx(B_t) \subseteq V(T^{glob}_{t-1})$ such that $tx^*_t \in Tx(B_t)$. On the other hand, the out degree of every transaction in $T^*_t$ is 0 and the out degree of every vertex in $\bar{T_t}$ is at least 1. The reason for the aforementioned constraints on the vertices of $G^{glob}$ and $T^{glob}_t$ is that when a block is found, block reward is created ``out of thin air'', and can hence be a designated as a transaction with no dependencies on which future transactions can depend. In addition, if $A \subseteq G^{glob}_t$, we let $Tx(A) = \cup_{B_j \in V(A)} Tx(B_j)$ be the set of all induced transactions from the subgraph $A$. Finally, these time-evolving graphs will have the property that if $t_1 < t_2$, then $G^{glob}_{t_1} \subseteq G^{glob}_{t_2}$ and $T^{glob}_{t_1} \subseteq T^{glob}_{t_2}$.

Let us now explore both the block and the transaction directed graphs from the perspective of a miner. We suppose that each miner $m_i$ has the following information at the end of turn $t$:
\begin{itemize}
    \item $G_{i,t}^{pub}$: The DAG consisting of all blocks $m_i$ has inferred from $G^{glob}_t$ via the P2P network.    
    \item $PB_{i,t} \subseteq V(G^{glob}_t)$: A set of private blocks $m_i$ has not yet shared to the P2P network.
    \item $T_{i,t}^{pub}$: The directed graph consisting of all transactions $m_i$ has inferred from $T^{glob}_t$ via the P2P network.
    \item $PT_{i,t} \subseteq V(T^{glob}_t)$: A set of private transactions $m_i$ has not yet shared to the P2P network.
\end{itemize}
Finally, we let $G^{pub}_t$ and $T^{pub}_t$ be the set of all blocks and transactions that have been shared to the P2P network respectively. 

\begin{definition}[Local Information]\label{def:local-info}
For a given miner $m_i$, we let $L_{i,t} = (G_{i,t}^{pub},PB_{i,t},T_{i,t}^{pub},\allowbreak PT_{i,t})$ and say this is the local information available to miner at the end of round $t$. We also say that $L_t = (L_{i,t})_{i=1}^t$ is the local information of all miners at the end of round $t$.
\end{definition}
 
 We conclude by defining what we mean by a single-step P2P information update for a miner, as well as what the strategy space available to a miner is.

\begin{definition}[Information Update] \label{def:inf-update}
Suppose that $H \subseteq G$ are graphs. Furthermore, suppose that the vertex set $A \in V(G) \setminus V(H)$. We define the distribution $U((H,G),A,q)$ as a single P2P information update via a specific sampling procedure. To sample $G' \sim U((H,G),A,q)$ we do the following:
\begin{itemize}
    \item Let $X = \emptyset$
    \item Independently, for each $v\in A$, with probability $q$, add $v$ to $X$.
    \item Let $G' = \Gamma(V(H) \cup X \mid G )$.
\end{itemize}

\end{definition}

\begin{definition}[Memoryless Miner Strategies] \label{def:miner-strat}
A miner strategy for $m_j$ is denoted by $S_j = (S^I_j,S^P_j, S^T_j)$ and consists of an initialisation strategy $S^I_j$, a publishing strategy $S^P_j$, and a transaction creation strategy $S^T_j$. Each of these functions takes as input $L_{j,t} = (G_{j,t}^{pub},PB_{j,t},T_{j,t}^{pub},PT_{j,t})$ at any given round $t$. 
\begin{itemize}
    \item  \emph{Initialisation strategy:} $S^I_j(L_{j,t}) = (X^I,Y^I)$ where $X^I \subseteq V(T_{j,t-1}^{pub}) \cup PT_{j,t-1}$ and $Y^I \subseteq V(G_{j,t-1}^{pub}) \cup PB_{j,t-1}$. Furthermore, $|X^I| \leq \eta $ and $1 \leq |Y^I| \leq  k$. 
    \item \emph{Publishing strategy:} $S^P_j(L_{j,t}) = (X^P,Y^P)$ where $X^P \subseteq PB_{j,t}$ and $Y^P \subseteq PT_{j,t}$. with the property that if $B_i \in X^P \Rightarrow tx^*_i \in Y^P$. 
    \item \emph{Transaction creation strategy:} $S^T_j(L_{j,t}) = \left( \{x_1,...,x_k\}, \{\Gamma^1(x_1),..., \Gamma^1(x_k)\}, W \right)$, where each $x_i \notin V(T^{priv}_{t-1})$, each set $\Gamma^1(x_i) \subseteq V(T^{priv}_{t-1})$ is non-empty, and $W \subseteq \{x_1,...,x_k \}$.
\end{itemize}

\end{definition}

To make sense of Definition \ref{def:miner-strat}, it suffices to note that $S^I_j$ is invoked when $m_j$ is chosen to mine a block. Set $X^I$ represents the set of the transactions that the block will contain. The number of these transactions can be at most $\eta$ and each block forcibly contains $tx^*_t$. On the other hand, set $Y^I$ describes the set of the blocks that the newly created block will point to. The number of these blocks can be at least 1 and at most $k$. Moving to $S^P_j$, this is invoked when $m_j$ wishes to publish hidden blocks/transactions to the P2P network. Finally, $S^T_j$ is invoked when $m_j$ wishes to create an arbitrary (finite) amount of new transactions $x_1,..,x_k$ that depend on transactions in $T^{priv}_{j,t-1}$ (each $x_i$ has a non-empty set $\Gamma^1(x_i)$ of dependencies). Notice that since $\Gamma^1(x_i) \neq \emptyset$, that forcibly each $x_i$ can not be of the form $tx^*_r$ for some $r$. Finally, $W \subseteq \{x_1,...,x_k\}$ represents which of the newly created transactions will be broadcast to the P2P network.


\subsection{Ledger Growth per State}\label{sec:lgps}

In order to describe and explain ledger growth in DAG-based ledgers in more detail, we introduce the notion of a ledger \emph{state} and describe how the ledger transitions from state to state. More specifically, at the end of round $t$, we say the ledger is in state $\sigma_t = \{L_t, G^{glob}_t,\allowbreak T^{glob}_t, G^{pub}_t, T^{pub}_t\}$. Recall that as we already mentioned in Section \ref{sec:lgp}, each turn $t$ consists of four different phases. We now present and describe formally what happens at each phase of a turn in terms of ledger growth.

\subsubsection{The Genesis Block and the Beginning of the Ledger: $\sigma_0$}
\label{subsubsec:genesis}
We bootstrap the ledger by creating a genesis block, $B_0$, which comes along with a genesis transaction $tx^*_0$. As such, $V(G^{pub}_0) = B_0$, $E(G^{pub}_0) = \emptyset$, $V(T^{pub}_0) = tx^*_0$, and $E(T^{pub}_0) = \emptyset$. In addition, $L_{0,i} = (G^{pub}_0, \allowbreak \emptyset,tx^*_0,\emptyset)$ for each $m_i$, $G^{pub}_0 = G^{glob}_0$, and $T^{pub}_0 = T^{glob}_0$, thus fully defining $\sigma_0$.

\subsubsection{The Mining Phase of Round $t$: $\Delta_1$}
\label{subsubsec:mining-phase}
Suppose that $\sigma_{t-1}$ is given. First a random miner is chosen, where the probability that each $m_i$ is chosen is precisely $h_i$. Suppose $m_j$ is drawn, the block $B_t$ is initialised as follows: $Tx(B_t) = \pi_1(S^I_j(L_{j,t})) \cup tx^*_t$ and we update the global block DAG by letting $V(G^{glob}_t) = V(G^{glob}_{t-1}) \cup B_t$ and $E(G^{glob}_t) = E(G^{glob}_{t-1}) \cup \{(B_t,y) \ \mid \ y \in \pi_2(S^{I}_j(L_{j,t})) \}$. In addition, for $r \neq j$, we let $L_{r,t-1}' = L_{r,t-1}$, and for $j$, we let $L_{j,t-1}' = ((G^{pub}_{j,t-1}, PB_{j,t-1}\allowbreak \cup B_t),(T^{pub}_{j,t-1}, PT_{j,t-1} \cup tx^*_t))$, so that $L_{t-1}' = (L_{i,t-1}')_{i=1}^n$. With this in hand, we can express the randomised transition $\Delta_1(\sigma_{t-1})$ as follows: $\Delta_1(\sigma_{t-1}) = \{ L_{t-1}', G^{glob}_t, T^{glob}_{t-1}, G^{pub}_{t-1}, \allowbreak T^{pub}_{t-1} \}
$

\subsubsection{The Information Update Phase of Round $t$: $\Delta_2$}
\label{subsubsec:info-update-phase}
Suppose that $\Delta_1(\sigma_{t-1}) =  \{ L_{t-1}', G^{glob}_t,\allowbreak T^{glob}_{t-1},  G^{pub}_{t-1}, T^{pub}_{t-1} \} $ is given after block initialisation. For each miner $m_i$, we sample $A_i \sim U \left( (G^{pub}_{i,t-1}, G^{pub}_{t-1}), V(G^{pub}_{t-1}) \setminus V(G^{pub}_{i,t-1}), q_i \right)$ from the block-DAG, and from the transaction DAG we sample $C_i \sim U \left( (T^{pub}_{i,t-1}, T^{pub}_{t-1}), V(T^{pub}_{t-1}) \setminus V(T^{pub}_{i,t-1}), q_i \right)$. We also let $F_i = \Gamma (V(C_i) \cup Tx(A_i) \mid T^{glob}_{t-1})$. Accordingly, we can define $L_{i,t}{''} = ((A_i, PB_{i,t-1}),(F_i,PT_{i,t-1}))$, and $L_{i,t-1}{''} = (L_{t-1}{''})_{i=1}^n$. With this in hand, we can express the randomised transition $\Delta_2 \circ \Delta_1 (\sigma_{t-1})$ as follows:
$\Delta_2 \circ \Delta_1 (\sigma_{t-1}) = \{ L_{t-1}{''},  G^{glob}_t, T^{glob}_{t-1}, G^{pub}_{t-1}, T^{pub}_{t-1} \}
$

\subsubsection{The Action Phase of Round $t$: $\Delta_3$}
\label{subsubsec:action-phase}
Let $\Delta_2 \circ \Delta_1 (\sigma_{t-1}) = \{ L_{i,t}{''},  G^{glob}_t, \allowbreak T^{glob}_{t-1}, G^{pub}_{t-1}, T^{pub}_{t-1} \}$ be given after an information update phase. We first update the global transaction graph by letting the vertices of a new graph be:
$$
V(T^{glob'}_t) = V(T^{glob}_{t-1}) \cup \{tx^*_t\} \cup \left( \cup_{i=1}^n \pi_2(S^P_i(L_{i,t-1}{''})) \right) \cup \left( \cup_{i=1}^n \pi_3( S^T_i(L_{i,t-1}{''}) ) \right) 
$$ 
and letting the edges of the new graph be:
$$
E(T^{glob'}_t) = E(T^{glob}_{t-1}) \cup \left( \cup_{i=1}^n E(\pi_2(S^P_i(L_{i,t-1}{''}))) \right) \cup \left( \cup_{i=1}^n E (\pi_3( S^T_i(L_{i,t-1}{''}) ) ) \right)
$$
where we slightly abuse notation and for a set of transactions, $A$, we denote $E(A)$ as the set of all tuples $(x,v)$ such that $x \in A$ and $v$ is a transaction $x$ depends on. With this in hand, We can define the following objects according to the deterministic strategies $S_i$ of each $m_i$:

\begin{itemize}
    \item $G^{pub}_{i,t} = \Gamma (V(G^{pub}_{i,t-1}) \cup \pi_1(S^P_i(L_{i,t-1}{''})) \mid G^{glob}_t)$
    \item $PB_{i,t} = PB_{i,t-1}{''} \setminus \pi_1(S^P_i(L_{i-,t-1}{''}))$
    \item $T^{pub}_{i,t} = \Gamma (V(T^{pub}_{i,t-1}) \cup \pi_2(S^P_i(L_{i,t-1}{''})) \mid T^{glob'}_t)$
    \item $PT_{i,t} = PT_{i,t-1} \setminus \pi_2(S^P_i(L_{i-,t-1}{''}))$
\end{itemize}
Finally, we let $L_{i,t} = ((G^{pub}_{i,t}, PB_{i,t}),(T^{glob}_{i,t}, PT_{i,t}))$, $L_t = (L_{i,t})_{i=1}^n$,  $G^{pub}_t = \cup_{i=1}^n G^{pub}_{i,t}$ and $T^{pub'}_t = T^{pub}_{t-1} \cup \left( \cup_{i=1}^n T^{pub}_{i,t} \right)$. With this in hand, we can express the transition $\Delta_3 \circ \Delta_2 \circ \Delta_1 (\sigma_{t-1})$:

$$
\Delta_3 \circ \Delta_2 \circ \Delta_1 (\sigma_{t-1}) = \{L_t, G^{glob}_t, T^{glob'}_t, G^{pub}_t, T^{pub'}_t \}
$$

\subsubsection{Nature Adds Transactions to $T^{glob}_t: \Delta_4$}
\label{subsubsec:nature-tx}
Finally, we model transactions that are naturally produced by agents other than miners, $m_1,...,m_n$ (the end-users of the ledger) and broadcast via the P2P network. To this end, we suppose that $\mathcal{D}$ is a pre-defined function, such that $\mathcal{D}(G^{glob}_t, T^{glob'}_t)$ is a distribution that returns $( \{x_1,...,x_k\}, \{\Gamma^1(x_1),...,\Gamma^1(x_k)\})$, a random finite set of new transactions $\{x_1,...,x_k\}$ along with their dependencies $\{\Gamma^1(x_1),...,\Gamma^1(x_k)\}$, with the property that $\Gamma^1(x_i) \neq \emptyset$ for each $x_i$. This latter point ensures that no $x_i$ is of the form $tx^*_r$. With this in hand, we let $A \sim \mathcal{D}(G^{glob}_t, T^{glob'}_t)$ and define $V(T^{glob}_t) = T^{glob'}_t \cup \pi_1(A)$, $E(T^{glob}_t) = E(T^{glob'}_t) \cup \pi_2(A)$, $V(T^{pub}_t) = T^{pub'}_t \cup \pi_1(A)$, and $E(T^{pub}_t) = E(T^{pub'}_t) \cup \pi_2(A)$. Ultimately, this allows us to define
$$
\Delta_4 \left( \{L_t, G^{glob}_t, T^{glob'}_t, G^{pub}_t, T^{pub'}_t \} \right) = \{L_t, G^{glob}_t, T^{glob}_t, G^{pub}_t, T^{pub}_t \}
$$

\subsubsection{Putting Everything Together.}In the previous sections we have detailed every process of the evolution of a DAG-based ledger. Before stating the formal definition in its full generality, we recall that the aforementioned growth process has several parameters: $\theta = \left( \vec{m}, \vec{h}, \vec{q}, \vec{S}, \eta, k, \mathcal{D} \right)$.

\begin{itemize}
    \item $\vec{m} = (m_1,...,m_n)$ are strategic miners.
    \item $\vec{h} = (h_1,...,h_n)$ are the hash rates of each miner.
    \item $\vec{q} = (q_1,...,q_n)$ is the informational parameter of each miner.
    \item $\vec{S} = (S_1,...,S_n)$ are ledger-building strategies for strategic miner
    \item $\eta \geq 1$ is the maximum number of transactions in a single block.
    \item $k \geq 1$ is the maximum out-degree of a block in the block DAG.
    \item $\mathcal{D}$ is a natural transaction generation distribution as per above.
\end{itemize}

To avoid confusion, we let $\Delta_\theta = \Delta_4 \circ \Delta_3 \circ \Delta_2 \circ \Delta_1$, with the influence of $\theta$ implicit in our previous exposition. With this we are able to fully formulate ledger growth in its complete generality. 

\begin{definition}[Formal Ledger Growth] \label{def:ledger-growth}
Suppose that $\theta = \left( \vec{m}, \vec{h}, \vec{q}, \vec{S}, \eta, k, \mathcal{D} \right)$. We let $\mathcal{P}_\theta$ denote DAG-based ledger growth governed by $\theta$ and define it recursively as follows: We let $\sigma_0$ be defined as per Section \ref{subsubsec:genesis}, and consecutively let $\sigma_t = \Delta_\theta^t(\sigma_0)$ for $t \geq 1$.

\end{definition}

\section{$\mathcal{P}_{f,k}$ Ledger Models and Honest Behaviour}\label{sec:hon}

In the previous section we described at length how a ledger grows under arbitrary parameters $\theta = \left( \vec{m}, \vec{h}, \vec{q}, \vec{S}, \eta, k, \mathcal{D} \right)$. The main purpose of this section is twofold: first we introduce a family of honest strategies that generalise honest mining in Bitcoin and SPECTRE called $\mathcal{P}_{f,k}$ mining, and second we introduce constraints on $\mathcal{D}$ that represent honest transaction generation by end-use agents in a DAG-based ledger (this includes Bitcoin and SPECTRE as well).   


\begin{definition}[Depth and Weight of a Block]\label{def:depth-weight}
Suppose that $G \in \mathcal{G}$ is a block-DAG. In other words, $G$ is connected and has a genesis block $B_0$. For a given $B_t \in G$, we let $w(B_t) = |\Gamma(\{B_t\} \mid G))| - 1$ be the weight of $B_t$. This is the number of predecessors $B_t$ has in $G$. We also define $D(B_t) = d_G(B_t,B_0)$ as the depth of $B_t$. This is the graphical distance between $B_t$ and $B_0$, i.e. the length of the (unique) shortest path between $B_t$ and $B_0$ in $G$. 
\end{definition}

In Bitcoin, miners resolve ambiguity in ledger consensus by initialising found blocks to point to the longest chain in the DAG. One reason for this is that agents have provably used significant computational power to grow said chain, and re-writing this history is thus computationally infeasible. In DAG-based ledgers, agents may point to multiple blocks. Thus, following this same thought process, they should point to blocks with a provably significant amount of computation in their histories. The issue, however, is that measuring how much computation exists in the past of a leaf is ambiguous in DAGs: a block could have either large weight or large depth (unlike in Bitcoin where these quantities are always the same), and it is unclear to decide which to give precedence to. In order to completely rank the importance of leaves in a block DAG, we simply use a family of score functions that expresses convex combinations of depth and weight.

\begin{definition}[Score Function]\label{def:score-function}
Suppose that $\alpha \in [0,1]$ and $\beta = 1- \alpha$. We say that $f$ is an $(\alpha,\beta)$ block-DAG score function if for a given block-DAG, $G \in \mathcal{G}$, $f(B_t) = \alpha D(B_t) + \beta w(B_t)$.
\end{definition}

In a nutshell, honest block-DAG growth in $\mathcal{P}_{f,k}$ protocols with parameter $\alpha$ and $\beta$ prescribes that miners prepare blocks with at most $k$ pointers that point the locally visible blocks in the block-DAG that with highest score under $f$. 

\subsection{Valid Blocks and Transactions}
\label{sec:valid-blocks-txs}

In ledgers employing decentralised consensus protocols, there is an explicit consensus mechanism whereby agents are able to look at their local view of the ledger and extrapolate valid blocks and subsequently valid transactions within the local view of the ledger. In Bitcoin for example, valid blocks consist of the longest chain in the ledger, and valid transactions consist of transactions within said longest chain. SPECTRE, on the other hand, has any seen block as valid, but the valid transaction extraction process is a complicated voting procedure that extracts a subset of transactions within the local view of the DAG as valid. We proceed by providing a definition of valid block and transaction extractors in $\mathcal{P}_{f,k}$ models that generalises both of these examples. 

\begin{definition}[Valid Block Extractors and Valid Transaction Extractors]
Suppose that $G$ is a block-DAG and that $\ell_1,..,\ell_k$ are the $k$ leaves in $G$ that have the highest score under $f$. Then we say that $VB(G) =  \Gamma (\{\ell_1,...,\ell_k\} \mid G)$ is the DAG of valid blocks in $G$ under $\mathcal{P}_{f,k}$. In addition, we let $VT(G) \subseteq Tx(VB(G))$ be the set of valid transactions for a specified transaction extractor function $VT$. We say that $VT$ is in addition monotonic if it holds that if $VB(H) \subseteq VB(G)$, then $VT(H) \subseteq VT(G)$.  

\end{definition}

In what follows we define a special type of monotonic valid transaction selection rule called \emph{present transaction selection}. The reason we outline this simple selection rule is that in Section \ref{sec:honest-tx-constitency} we will show that if all miners employ monotonic valid transaction selection and the honest strategies presented in Section \ref{sec:honest-strategy}, then we can assume without loss of generality that they employ present transaction selection as a valid transaction selection rule.

\begin{definition}[Present Transaction Selection]
Suppose that $G$ is a block-DAG and that $\ell_1,..,\ell_k$ are the leaves in $G$ that have the highest score under $f$. Then we say that $VB(G) =  \Gamma (\{\ell_1,...,\ell_k\} \mid G)$ is the DAG of valid blocks in $G$ under $\mathcal{P}_{f,k}$. In addition we say that $PVT(G) = Tx(VB(G))$ is the set of present valid transactions in $G$ under $\mathcal{P}_{f,k}$.

\end{definition}

\subsection{Defining $S^I, S^P$, and $S^T$ for Honest Mining in $\mathcal{P}_{f,k}$}
\label{sec:honest-strategy}

We define $H_j = (H^I_j,H^P_j,H^T_j)$ as the honest strategy employed by $m_j$ in $\mathcal{P}_{f,k}$, and describe each component below. 

\begin{itemize}
    \item $H^I_j$: Compute $A = VB(G^{pub}_{j,t})$ and $B = VT(G^{pub}_{j,t})$. Let $H^I_j(L_{j,t}) = (X,Y)$. $X$ is the set of at most $\eta$ oldest non-block-reward (i.e. not of the form $tx^*_r$) transactions in $T^{pub}_{j,t} \setminus B$ (ties are broken arbitrarily) with a graphical closure in $B$. $Y=\{\ell_1,...\ell_k\}$ is the set of $k$ highest-score leaves in $G^{pub}_{j,t}$ under $f$. 
    
    \item $H^P$: Publish all private blocks and transactions immediately 
    \item $H^T$: Create no new transactions (the assumption is that transactions created by pools are negligible with respect to the total transaction load of the ledger)
\end{itemize}

Before continuing, we note that in the $\mathcal{P}_{f,k}$ model, $H^T$ ensures that honest miners do not create and broadcast any transactions themselves. This, of course, is not the case in practice, but it is an accurate approximation to a regime in which the fraction of transactions created by miners is a negligible fraction of all transactions created by end-users of the ledger.

\myheader{Implementation of Bitcoin and SPECTRE as $\mathcal{P}_{f,k}$ Protocols} With the previous machinery in place, we can see that block-DAG and transaction-DAG growth in Bitcoin and SPECTRE are special cases of $\mathcal{P}_{f,k}$ ledgers. For Bitcoin, we let $k = 1$, and any parameter setting, $(\alpha,\beta)$ for $f$ results in Bitcoin growth. As for SPECTRE, we let $k = \infty$ and once more any parameter setting $(\alpha,\beta)$ for $f$ suffices to implement honest SPECTRE ledger growth.

\subsection{Honest Transaction Consistency}
\label{sec:honest-tx-constitency}

As mentioned in Section \ref{sec:valid-blocks-txs}, we can show that amongst monotonic transaction extractors, present transaction extractors are all we need for honest ledger growth in the $\mathcal{P}_{f,k}$ model.

\begin{theorem}
If the valid transaction extractor, $VT$, is monotonic and all miners employ $H = (H^I, H^P, H^T)$, then $VT$ is a present transaction extractor.
\end{theorem}

\begin{proof}
Suppose that $L_{i,t} = (G^{pub}_{i,t}, PB_{i,t},T^{pub}_{i,t}, PT_{i,t})$ is the local information available to $m_i$ at turn $t$. Since $m_i$ is honest, one can easily see that $G^{pub}_{i,t} =  G^{priv}_{i,t} = G_{i,t}$ and $T^{pub}_{i,t} = T^{priv}_{i,t} = T_{i,t}$. Clearly $VT(G_{i,t}) \subseteq PVT(G_{i,t})$. Now suppose that $x \in PVT(G_{i,t})$. This means that $x \in Tx(B_r)$ for some block $B_r$ found by say $m_j$. This means that in turn $r$, $m_j$ invoked $H^I$ to create $B_r$, which means that since $x \in Tx(B_r)$, all dependencies of $x$ are in $VT(\Gamma(B_r\mid G_{i,t}))$, the valid transactions from the DAG consisting of the closure of $B_r$ in the block DAG. However $VT(\Gamma(B_r\mid G_{i,t})) \subseteq VT(G_{i,t})$ since $VT$ is monotonic. Therefore $x$ has its dependencies met in $VT(G_{i,t})$ so that $x \in VT(G_{i,t})$. This implies $VT(G_{i,t}) = PVT(G_{i,t})$ as desired.  
\qed
\end{proof}

In light of this theorem, we focus on monotonic valid transaction extractors given their generality. Hence, from now on we assume that when we invoke $VT$, we in fact mean that $VT$ is a present transaction extractor.

\subsection{Honest Natural Transaction Generation}

Notice that $H^T$ dictates that each $m_i$ does not produce or propagate transactions created by themselves. Hence, it is crucial that we properly define $\mathcal{D}$ in the $\mathcal{P}_{f,k}$ model. At first one may be tempted to simply treat the random growth of $T^{glob}_t$ as independent of $G^{glob}_t$, but this is a grave mistake. To see why, imagine that $G^{glob}_t$ contains some block $B_r$ that is orphaned by each $m_i$ (note that this can only happen if $k < \infty$). If the growth of $T^{glob}_t$ is independent of that of $G^{glob}_t$, then it could be the case that many (if not infinitely many) future transactions depend on $t^*_r$. However, if $B_r$ is orphaned by all miners, $tx^*_r$ is not valid, hence none of these future transactions will be added to the ledger via close inspection of how $H^I$ is defined. 

A compelling fact is that if all miners have orphaned $B_r$, then chances are that whatever local view of $G^{pub}_t$ an end-user has, they too will have orphaned $B_r$, and thus will not have $tx^*_r$ as a valid transaction. In more direct terms, any money created via the block reward of $B_r$ is not actually in the system for an end-user, so if this end-user is honest, there is no reason why they would produce transactions that would depend on this illegitimate source of currency.

\begin{definition}[Honest Transaction Distributions] 
Let $G^{glob}_t$ be a global block DAG at turn $t$ with $k$ highest leaves are $\ell_1,...,\ell_k$. In an honest setting, $VB(G^{glob}_t) = \Gamma( \{\ell_1,...,\ell_l\} \mid G^{glob}_t) $ and $VT(G^{glob}_t) = Tx(VB(G^{glob}_t))$.  We say that $\mathcal{D}(G^{glob}_{t}, T^{glob}_t)$ is an honest transaction distribution if $x \sim \mathcal{D}(G^{glob}_{t}, T^{glob}_t)$ is such that $x \notin T^{glob}_t$ and its dependencies lie strictly in $VT(G^{glob}_t)$. 
\end{definition}

\subsection{Non-Atomic Miners}

It is common in the analysis of Bitcoin and related cryptocurrencies to group an arbitrary number of honest miners into one honest miner with the aggregate hash power of all different honest miners. The behaviour of multiple honest miners or one aggregate honest miner is indistinguishable from the local perspective of a single strategic miner. In the DAG-based ledger one could perform a similar analysis, however the partial information inherent in our model does not allow us to directly do so. The reason for this is that separate honest miners may have different partial views of the block DAG and transaction DAG. However, if for a given time horizon $t = 1,...,T$, it is the case that a group of honest miners each have a small enough hash power that they most likely never see more than at most one block in the time horizon, we can aggregate all said miners into one large honest miner that simply re-samples their view of the block DAG and transaction DAG whenever they are chosen to initialise a block and invoke $H^I$. We call such miners non-atomic and formalise their definition below.

\begin{definition}[Non-atomic miner]
For a given time-horizon, $t = 1,...,T$, we say that a group of miners $m_1',...,m_k'$ is non-atomic if with high probability, each such miner finds at most one block in this time horizon. 
\end{definition}

\myheader{Simulating Non-Atomic Miners}
When a group of non-atomic miners finds a block they need to create a fresh sample of what partial information they have in the block/transaction DAG. To do so, let us suppose that $m_1',...,m_k'$ are a group of non-atomic miners which we group together into a single miner, $m^*$ of hash power $\sum_{i=1}^k h_i$. Furthermore, we suppose that each non-atomic miner in this group has informational parameter $q$. When $m^*$ finds a block at time-step $t = 1,...,T$, we simply take each block/transaction present in the global block/transaction DAG, and for each turn it has been present in its respective global DAG, flip a biased coin of weight $q$ to see whether $m^*$ directly sees this block/transaction. Once this preliminary list of seen blocks/transactions is created, $m^*$ also sees all ancestors of said blocks/transactions.

\subsection{Payoffs and Transaction Generation Rate}

\myheader{Block Rewards and Transaction Fees}We suppose that at time-step $T$, miners get a normalised block reward of 1 per block that they have in $VB(G^{glob}_T)$. As for transaction fees, the full generality of $P_{f,k}$ protocols only specifies how to extrapolate valid transactions conditional upon everyone being honest, and not who receives transaction fees (this is subsumed in the details of $VT$ in the general setting). For this reason we further assume that transaction fees are negligible in comparison to block rewards over the time horizon $t = 1,...,T$.

\myheader{Transaction Generation Rate} Although in full generality there is no restriction on how many transactions nature may create in a given turn, we impose a fixed constraint on this quantity: $\lambda$. As such, each turn introduces $\{x_{t,1},...,x_{t,\lambda} \}$ transactions sampled from a specified honest transaction distribution $\mathcal{D}$. Furthermore, in our simulations we let $\lambda = \eta$, so that the ledger infrastructure can, in theory, cope with the transaction load if all miners have full information, and thus we can see specifically it falls short of this objective in the partial information setting.



\section{Computationally Modelling Honest Growth in $\mathcal{P}_{f,k}$ Ledgers}\label{sec:CMHG}

\begin{algorithm} 
\caption{Honest Ledger Growth}
\label{alg:honest_ledger_growth}                           
\begin{algorithmic}                    
    \STATE
	\item \textbf{Require:}
    \STATE Ledger growth Parameters: $\theta = \left( \vec{\alpha}, \vec{q}, \eta, k, \mathcal{D}, \lambda, T \right)$
    \STATE
	\item \textbf{Genesis:}
	\STATE $G_0 \leftarrow (\{B_0\}, \ \emptyset)$, \ $T_0 \leftarrow (\{tx^*_0\}, \emptyset), \ Tx(B_0) \leftarrow tx^*_0$	
	\FOR{$j = 1,\ldots,N$}
    	\STATE $G^{vis}_{i,0} \leftarrow G_0$  
    	\STATE $T^{vis}_{i,0} \leftarrow T_0$
	\ENDFOR
	\STATE 

        \FOR{$t = 1,\ldots,T$}
            \item \textbf{Mining Phase:}
            \STATE Sample $i$ from $\vec{\alpha}$
            \STATE Compute $\ell_1,...,\ell_r$, top $r \leq k$ leaves of $G^{vis}_{i,t-1}$ in terms of score, $f$
            \STATE Compute $tx_{i_1},...,tx_{i_s}$, with $s \leq \eta$, oldest pending transactions of non-zero out-degree in $T^{vis}_{i,t-1}$ (ties broken lexicographically)    
            \STATE $G_t \leftarrow (V(G_{t-1}) \cup B_t, \ E(G_{t-1}) \cup \{ (B_t,\ell_1),...,(B_t,\ell_r) \})$
            \STATE $Tx(B_t) = tx_t^* \cup \{tx_{i_1},...,tx_{i_s}\}$
            \STATE $Owner(B_t) = i$
            \STATE $V(G^{vis}_{i,t-1}) \leftarrow B_t, \ V(T^{vis}_{i,t-1}) \leftarrow tx^*_t$
            \STATE
            \item \textbf{Tx Generation:}
            \STATE Compute valid transactions: $VT_{t-1} \subseteq V(T_{t-1})$
            \STATE $T_t \leftarrow (V(T_{t-1}) \cup \{tx^*_t\}, \ E(T_{t-1})) $
            \FOR{$j = 1,\ldots,\lambda$}
                \STATE Use $\mathcal{D}$ and $VT_{t-1}$ to draw $\{tx_{i_1},...,tx_{i_s}\}$ dependencies of $tx_{t,j}$
                \STATE $V(T_t) \leftarrow V(T_t) \cup \{tx_{t,j}\}, \ E(T_t) \leftarrow E(T_t) \cup \{(tx_{t,j}, tx_{i_1}),..., (tx_{t,j}, tx_{i_s})\}$ 
                \STATE 
            \ENDFOR
            \item \textbf{Information Update:}
            \FOR{$j = 1,\ldots,N$}
                \STATE $G^{vis}_{j,t} \leftarrow G^{vis}_{j,t-1}, \ T^{vis}_{j,t} \leftarrow T^{vis}_{j,t-1}$
                \FOR{each $B_r \in V(G_{t}) \setminus V(G^{vis}_{j,t})$}
                    \STATE With probability $q_i$: $V(G^{vis}_{j,t}) \leftarrow V(G^{vis}_{j,t}) \cup \{B_r\}$
                \ENDFOR
                \FOR{each $tx_s \in V(T_t) \setminus V(T^{vis}_{j,t-1})$}
                    \STATE With probability $q_i$: $V(T^{vis}_{j,t}) \leftarrow V(T^{vis}_{j,t}) \cup \{tx_s\}$
                \ENDFOR
                \STATE $G^{vis}_{j,t} \leftarrow \Gamma(V(G^{vis}_{j,t}),G_t)$
                \STATE $T^{vis}_{j,t} \leftarrow \Gamma(Tx(G^{vis}_{j,t}) \cup V(T^{vis}_{j,t}), T_t)$
            \ENDFOR
        \ENDFOR   
    \STATE
    \RETURN $G_T, T_T, \{Owner(B_t)\}_{t = 1}^T, \{Tx(B_t)\}_{t=1}^T, \{G^{vis}_{i,T}\}_{i=1}^N, \{T^{vis}_{i,T}\}_{i=1}^N$
  
\end{algorithmic}
\end{algorithm}

We now describe the pseudo-code for honest ledger growth in $\mathcal{P}_{f,k}$ ledgers. In the implementation described in Algorithm \ref{alg:honest_ledger_growth} we make the following assumptions:
\begin{itemize}
    \item The underlying ledger is $\mathcal{P}_{f,k}$
    \item All miners are Honest
    \item $VT$ is a present transaction extractor
    \item The number of transactions created by miners is negligible with respect to the total number of transactions created within the P2P network. 
    \item The honest transaction generation distribution, $\mathcal{D}(G^{glob}_t,T^{glob})t)$, is modelled as follows: for a given transaction, $tx$, to be sampled via $\mathcal{D}$, first $k \sim Poiss(\gamma)$ is drawn as the number of dependencies $tx$ will have from $VT(G^{glob}_t)$, and subsequently, $x_1,...,x_k$ are drawn from $VT(G^{glob}_t)$ uniformly randomly, where we allow repetitions. Ultimately we let $\{y_1,..,y_r\}$ be the unique transaction set extrapolated from $\{x_1,...,x_k\}$ and this is precisely the set of dependencies of the newly sampled $tx$. We note that transactions can depend on the block reward of arbitrarily old blocks. Hence uniformity is not an unfeasible assumption in sampling dependencies of a new transaction given our time horizons. Finally, in our simulations, since $T = 50$, we let $\gamma = 2$ so that transactions on average have multiple dependencies on average.   
\end{itemize}
Given these assumptions, we also condense the notation used for key variables needed for ledger growth in the pseudo-code for Algorithm \ref{alg:honest_ledger_growth}. In particular:
\begin{itemize}
    \item $G_t$ and $T_t$ are the global block and transaction DAGs respectively
    \item $G^{vis}_{i,t}$ and $T^{vis}_{i,t}$ are the portions of $G_t$ and $T_t$ visible to $m_i$ at the end of turn $t$.
    \item $Owner(B_t) = i$ means that $m_i$ mined block $B_t$
\end{itemize}

\section{Results}

\subsection{Fairness} We recall that one of the key properties of Bitcoin is that it is fair: miners earn block reward proportional to the computational resources they expend on extending the ledger. One of the most significant observations from our simulations is that $\mathcal{P}_{f,k}$ ledgers are rarely fair as soon as agents begin having informational parameters, $q < 1$, as is the case in a high throughput setting. To illustrate this phenomenon, we study a two-miner scenario with agents $m_0$ and $m_1$ of hash power $(1-h_1,h_1)$ and informational parameters $(q_0,q_1)$. $m_0$ is modelled as a non-atomic miner and we empirically compute the surplus average block reward of $m_1$ relative to the baseline $h_1$ they would receive in a fair protocol. Our results are visualised in Figure \ref{fig:highres_surplus}. Each row of the figure represents $k=1,2,3$ respectively and each column represents $q_0 = 0.005, 0.05, 0.2$. Each individual heatmap fixes $k$ and $q_0$ and plots average block reward surplus for $m_1$ as $q_1 \in [0,1]$ and $h_1 \in (0,0.5]$ are allowed to vary. Finally, each pixel contains the average block reward surplus for  $T=50$ and averaged over 50 trials. We notice that an added strength to our fairness result is that they hold, irrespective of the underlying honest transaction distribution $\mathcal{D}$ used in practice. 

The most jarring observation is that that for a large amount of parameter settings, $m_1$ earns a vastly different average block reward than their fair share $h_1$. In fact, for fixed $k$ and $q_0$, there seem to be three regions of the hash space $h_1 \in (0,0.5]$ with qualitatively distinct properties:
\begin{itemize}
    \item If $h_1$ is large enough, $m_1$ strictly benefits from having lower $q_1$ values. This is due to the fact that an honest miner with small $q_0$ necessarily sees his own blocks and is inadvertently acting somewhat ``selfishly''. Hence if their hash rate is high enough, their persistent mining upon their own blocks may end up orphaning other blocks and give them a higher share of valid blocks in the final DAG.
    \item If $h_1$ is small enough, $m_1$ strictly benefits from having higher $q_1$ values. Contrary to the previous point, at small hash values, $m_1$ only finds a few blocks, and hence they risk losing their entire share of blocks if these blocks aren't well positioned in the block DAG, since they are in no position to inadvertently overtake the entire DAG via pseudo-selfish behaviour resulting from low $q_1$ values. 
    \item Finally, for intermediate $h_1$ values, $m_1$ no longer has a monotonic surplus with respect to $q_1$ but rather a concave dependency. This can be seen as an interpolation of the previous two points. 
\end{itemize}
We notice that where these qualitative regions of $h_1$ values lie within $(0,0.5]$ depends entirely on $k$ and $q_0$. In general, for fixed $k$ (i.e specific rows within Figure \ref{fig:highres_surplus}), as $q_0$ increases, the transitions between these regions shift rightwards, and for fixed $q_0$ (i.e. specific columns in Figure \ref{fig:highres_surplus}) as $k$ increases, also shifts rightwards, as increasing $k$ can be seen to informally have the same effect as uniformly increasing $q_0$ and $q_1$ as agents are more likely to see blocks due to multiple pointers. 

\subsection{Hash Power / Information Tradeoff} We note that in the high throughput setting, miners may be faced with the choice of investing resources into increasing their connectivity to the underlying P2P network of the protocol or their hash power. The incentives behind such a decision are ultimately governed by how much it costs for $m_1$ to improve either $q_1$ or $h_1$. Figure \ref{fig:highres_surplus} clearly shows that such a decision is non-trivial. Roughly speaking ``small'' miners benefit from increasing their connectivity to the P2P network and ``large'' miners may even benefit more from having less connectivity to the P2P network!

Finally we would like to mention that for $k = \infty$ there is no point in plotting fairness as per our definitions, since such a value of $k$ automatically makes $\mathcal{P}_{f,k}$ fair. On the other hand, as we will see shortly, even when $k = \infty$, $\mathcal{P}_{f,k}$ protocols suffer from efficiency shortcomings at high enough transaction loads. 

\begin{figure}[h]
    \centering
    \includegraphics[width=.32\linewidth]{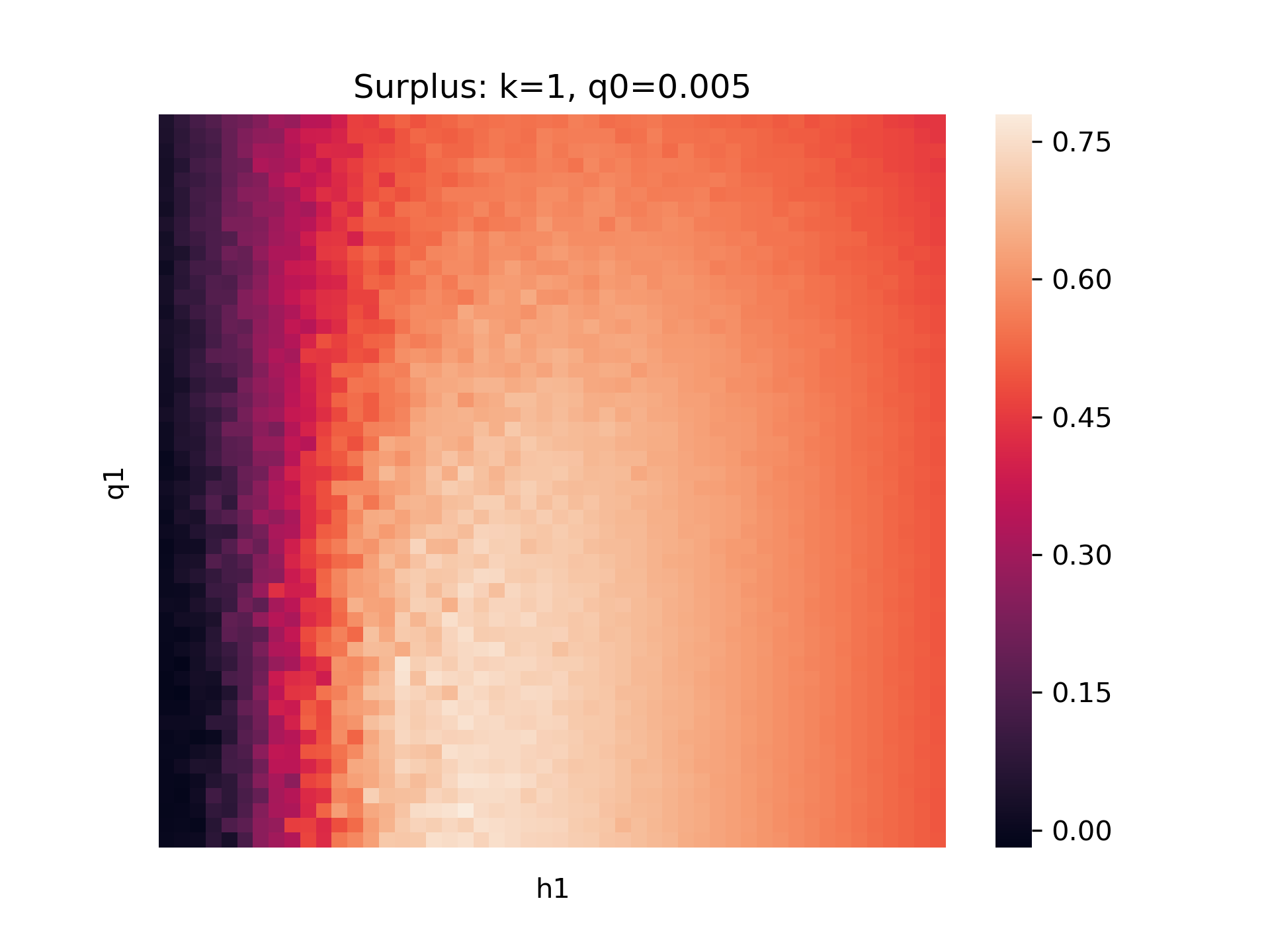}
    \includegraphics[width=.32\linewidth]{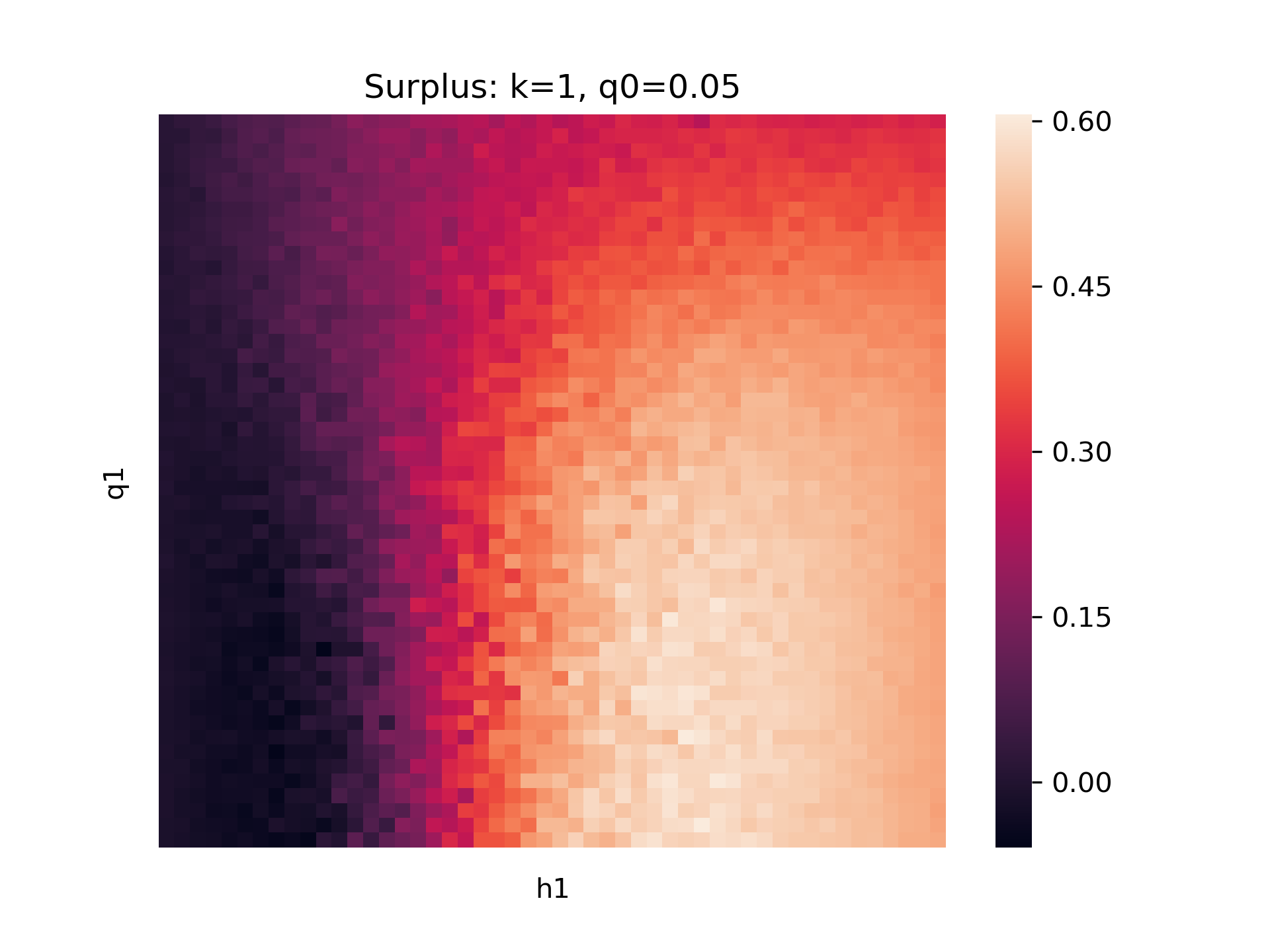}
    \includegraphics[width=.32\linewidth]{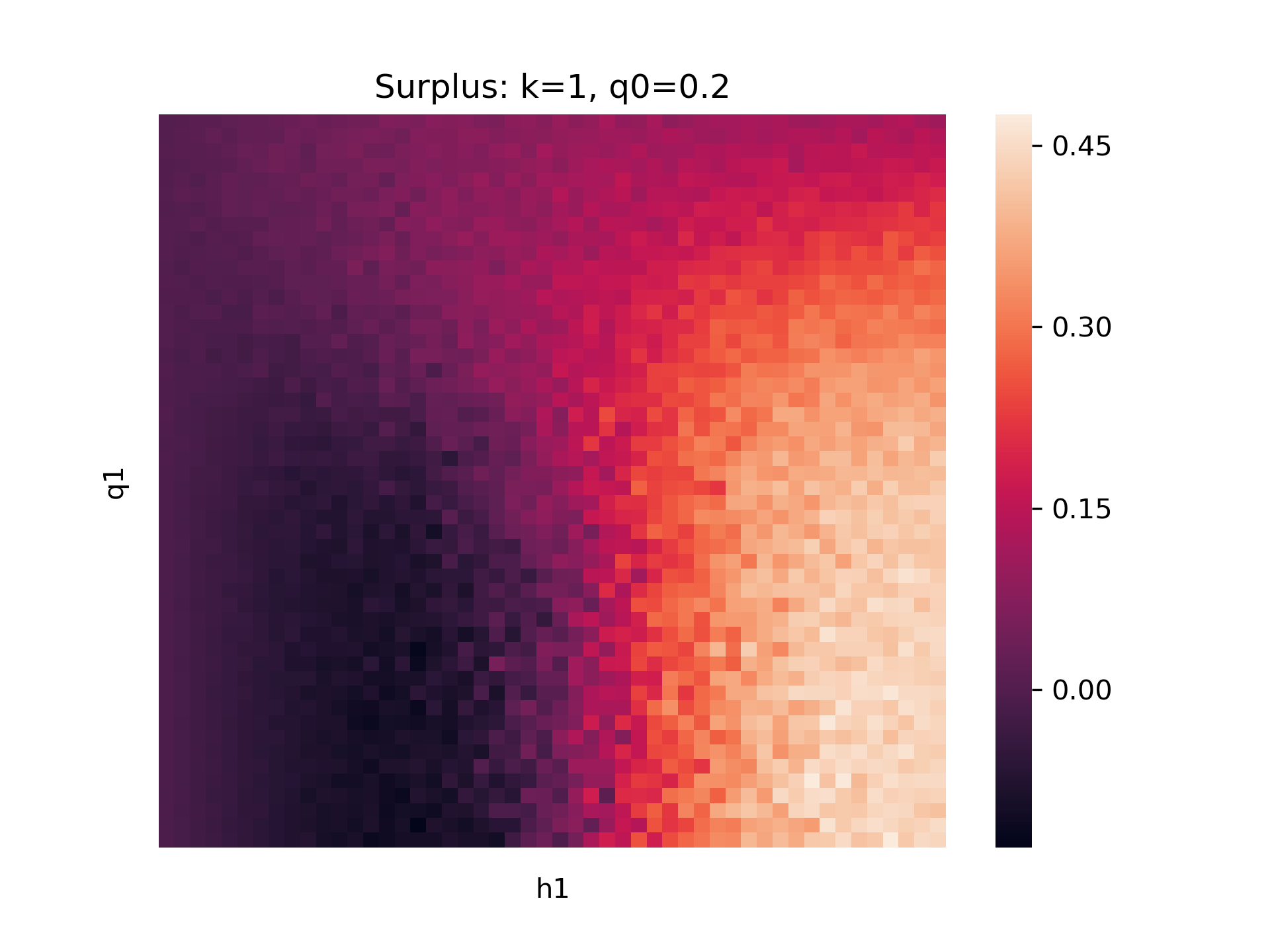}
    \includegraphics[width=.32\linewidth]{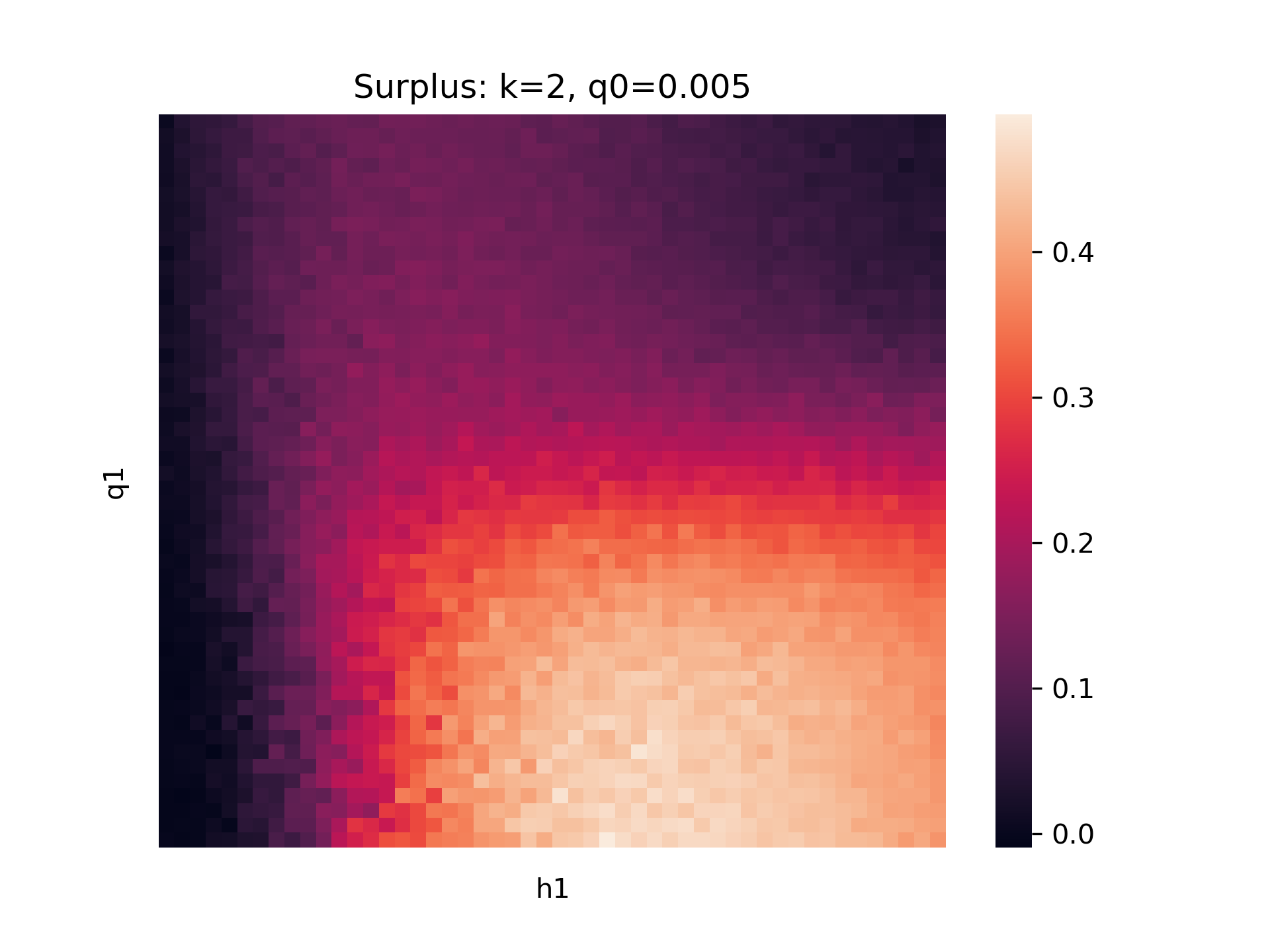}
    \includegraphics[width=.32\linewidth]{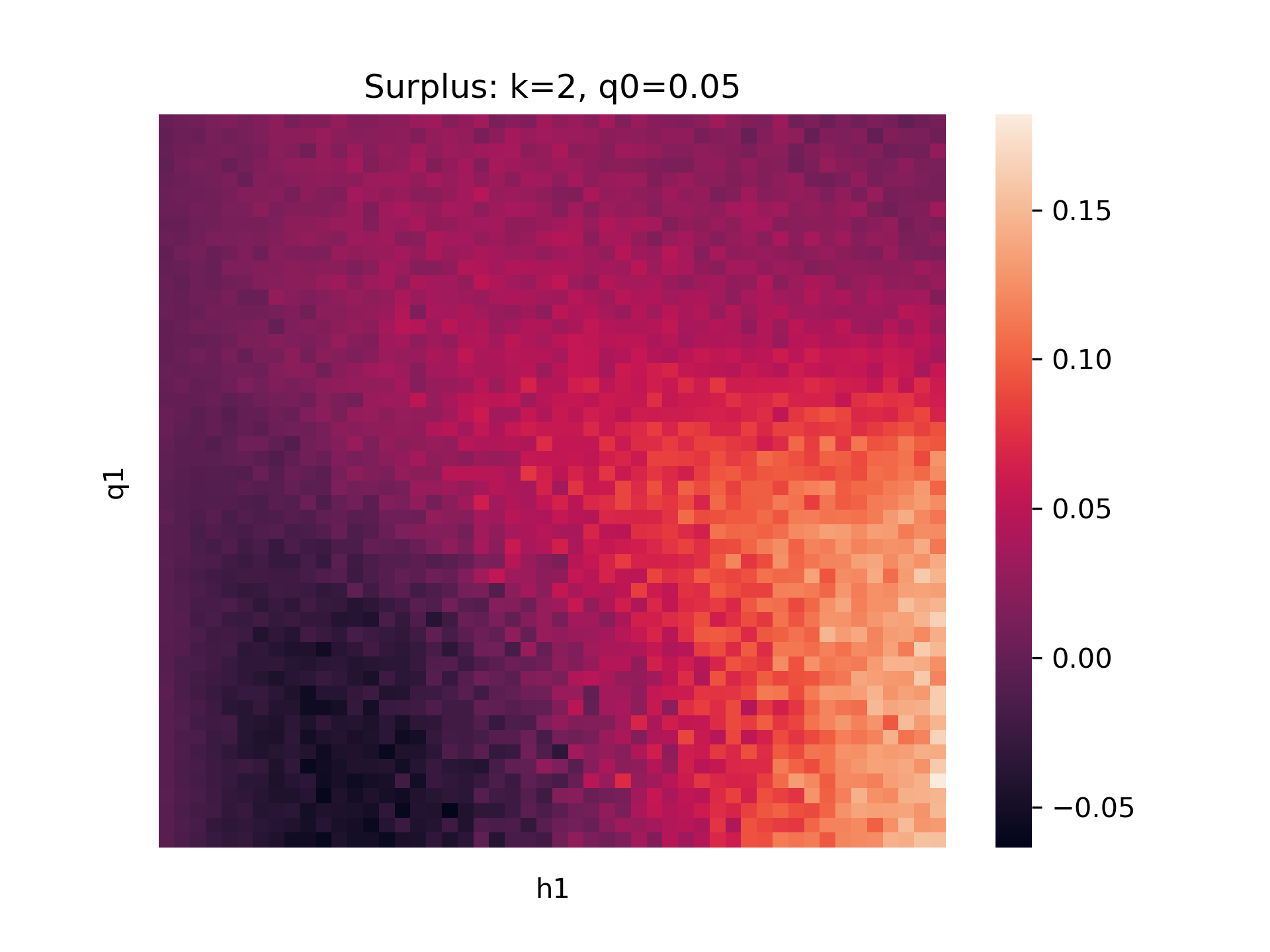}
    \includegraphics[width=.32\linewidth]{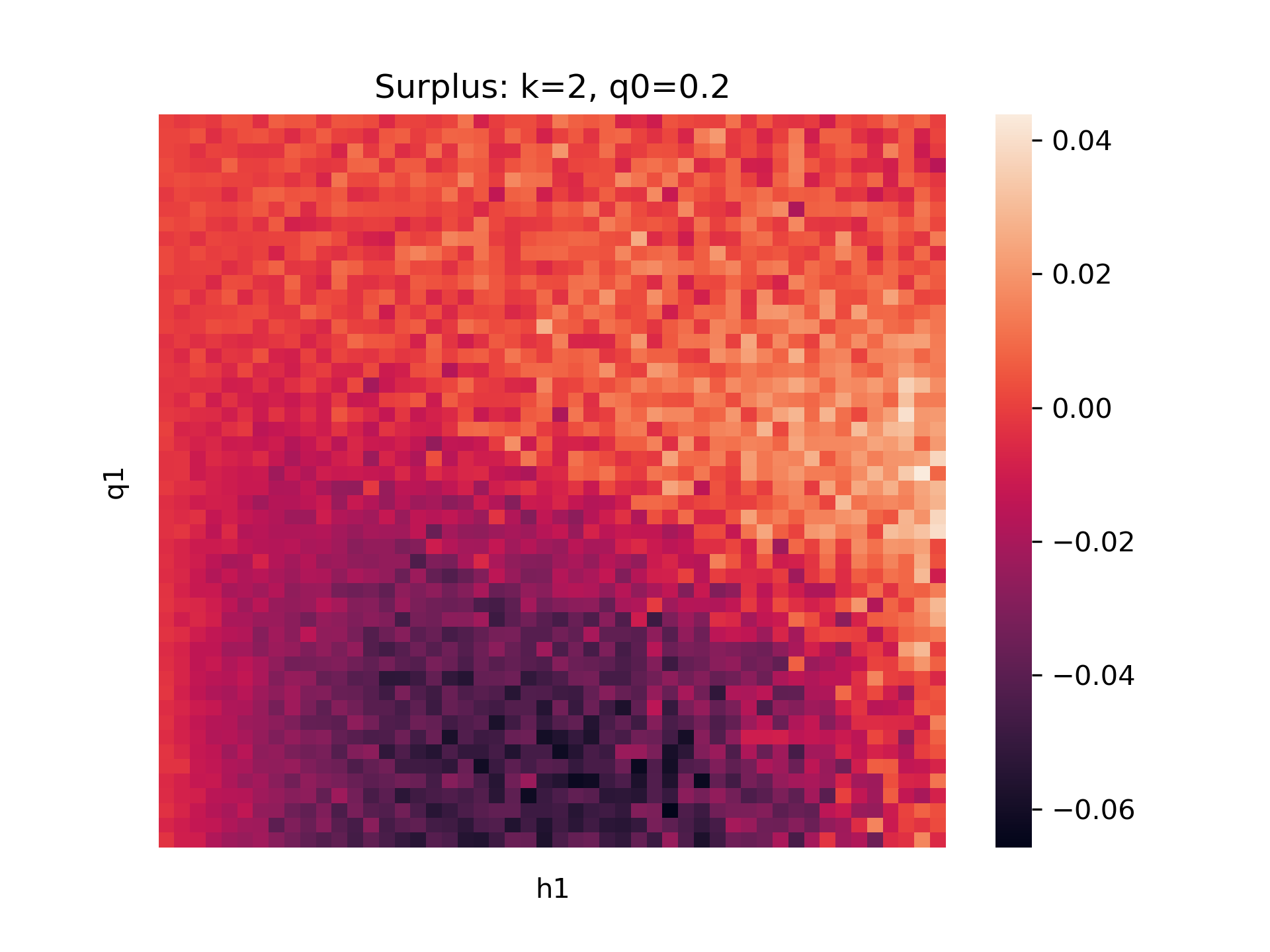}
    \includegraphics[width=.32\linewidth]{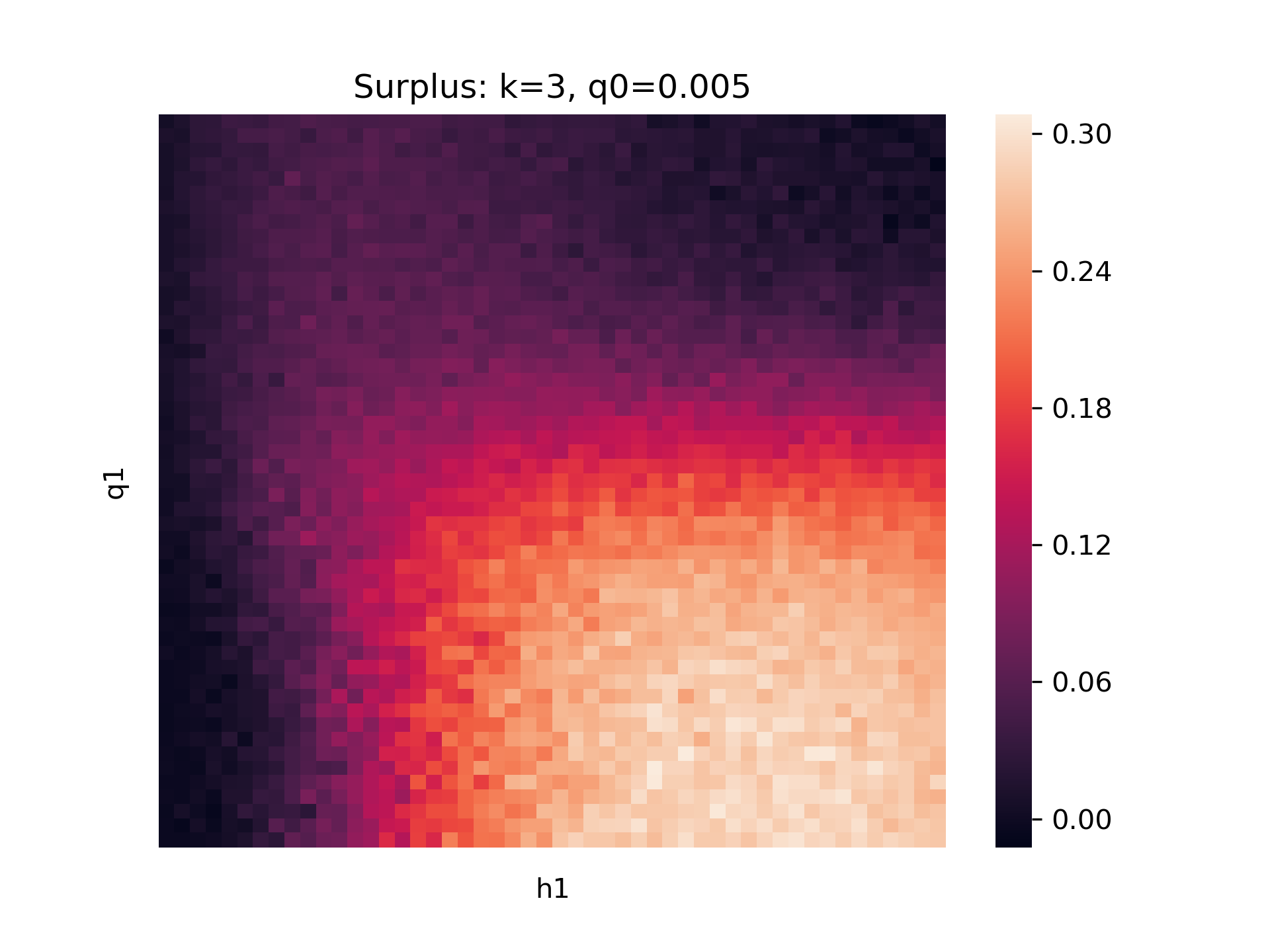}
    \includegraphics[width=.32\linewidth]{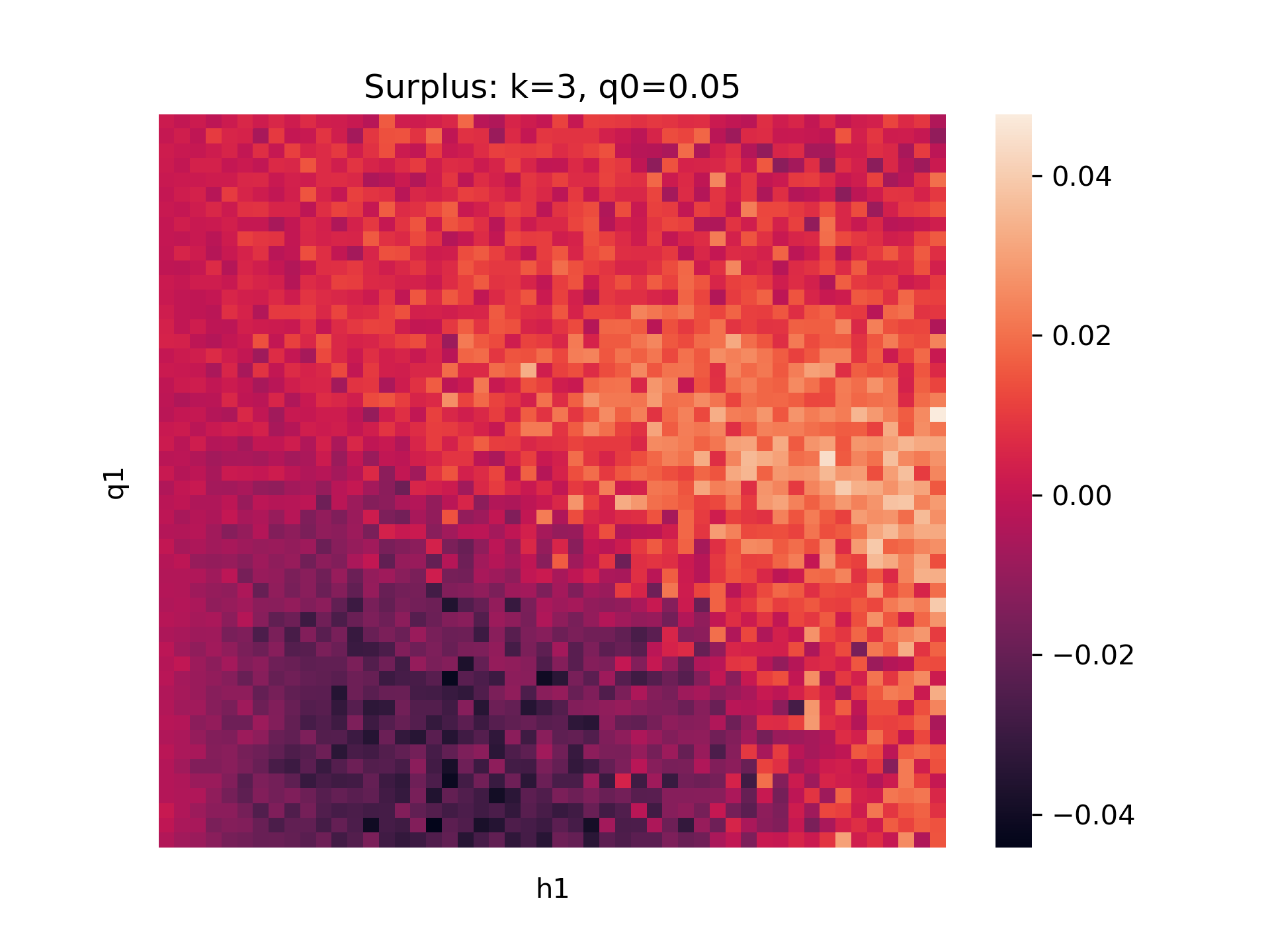}
    \includegraphics[width=.32\linewidth]{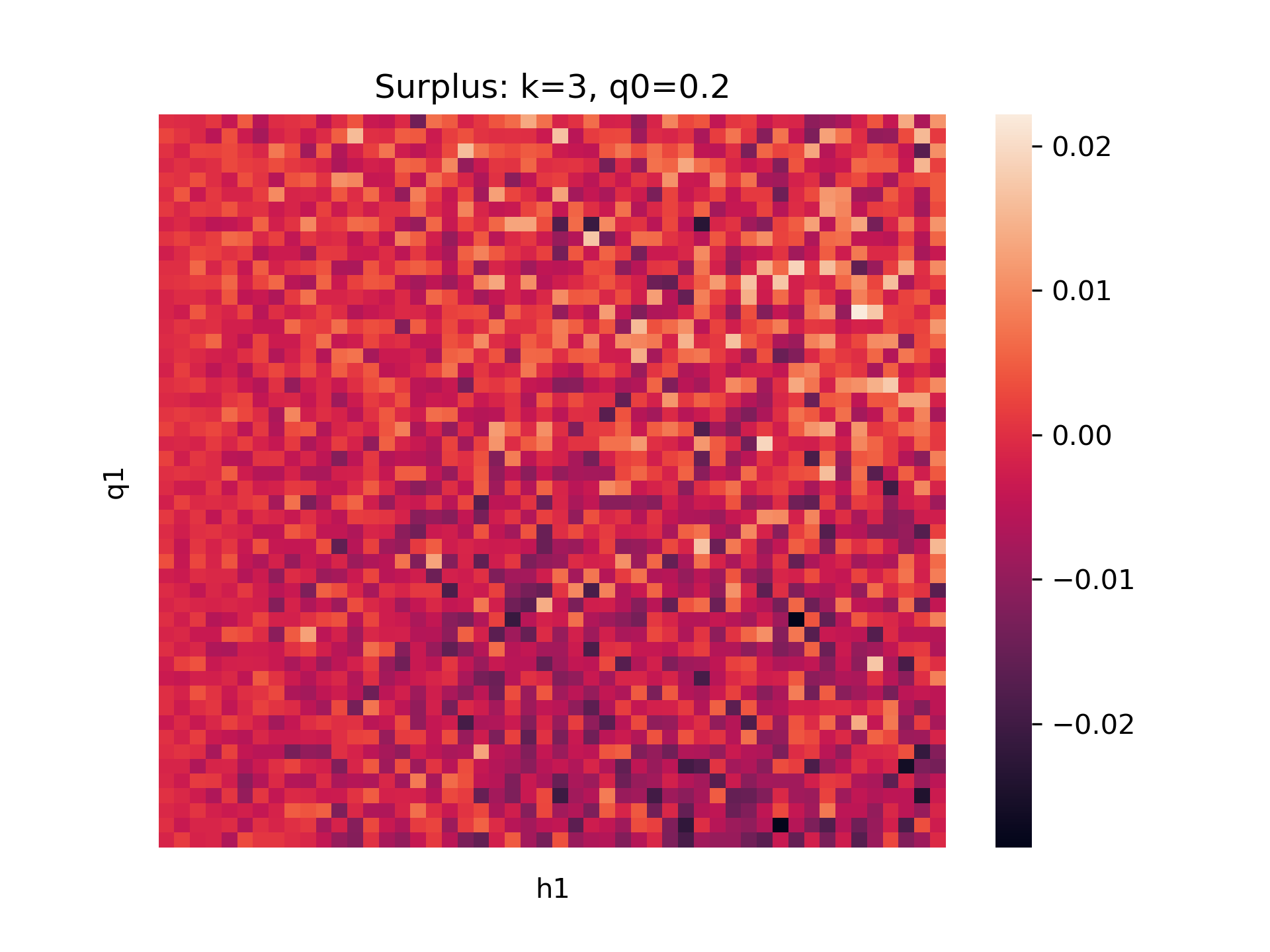}
    \caption{Surpluses for $k=1,2,3$ at $q_0 \in \{0.005, 0.05, 0.2\}$ }
    \label{fig:highres_surplus}
\end{figure}

\subsection{Ledger Efficiency}DAG-based ledgers have been created with the aim of tackling a higher transaction load in cryptocurrencies. Given that we have a way of modelling honest transaction growth, there are three different metrics we use to precisely quantify how well DAG-based ledgers deal with a higher throughput of transactions. The first and most important is the \emph{Proof of Work Efficiency}.
More specifically, for a given DAG-based Ledger, we say that the PoW efficiency is the fraction of globally valid transactions that are present within the valid sub DAG of the block DAG, over all published transactions.

This is the most important metric, since the goal of a ledger is to maximise the rate at which new transactions are processed. We also compare ledgers in terms of the average fraction of orphaned blocks they create and their transaction \emph{lag}, which is defined as the time difference between the issue and successful inclusion of the DAG's most recent transaction and the final turn of the time horizon.

For our experiments, we compared $\mathcal{P}_{f,k}$ performance for $k \in \{1, 2, 3, \infty\}$ and $n=4$ atomic miners each with $h_i = 1/4$ and varying $q_i$'s. 
For all graphs, we have $\eta =  6, T=100$ and the results have been averaged over 50 trials.
\begin{figure}[h]
    \centering
    \includegraphics[width=.34\linewidth]{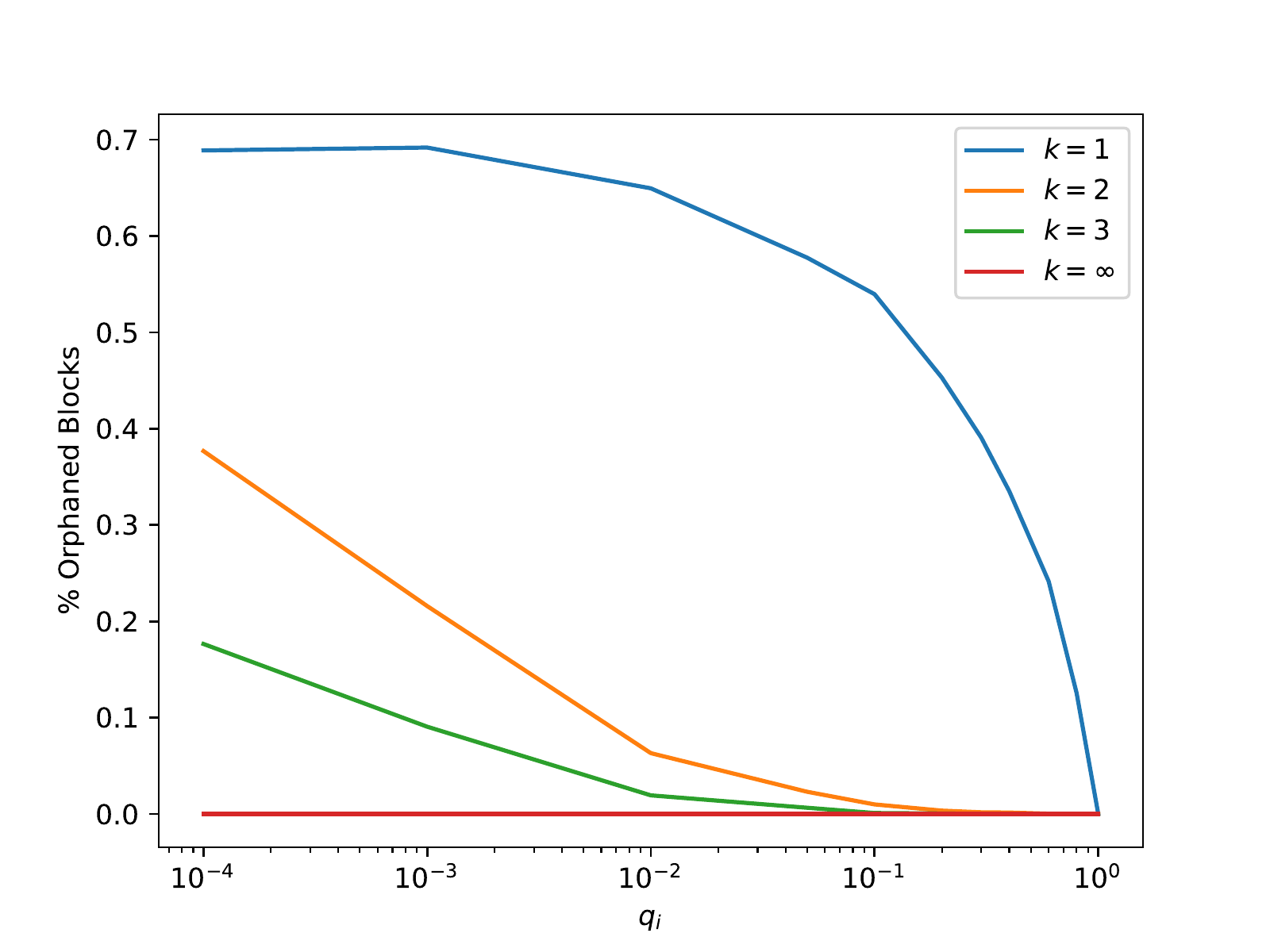}
    \includegraphics[width=.32\linewidth]{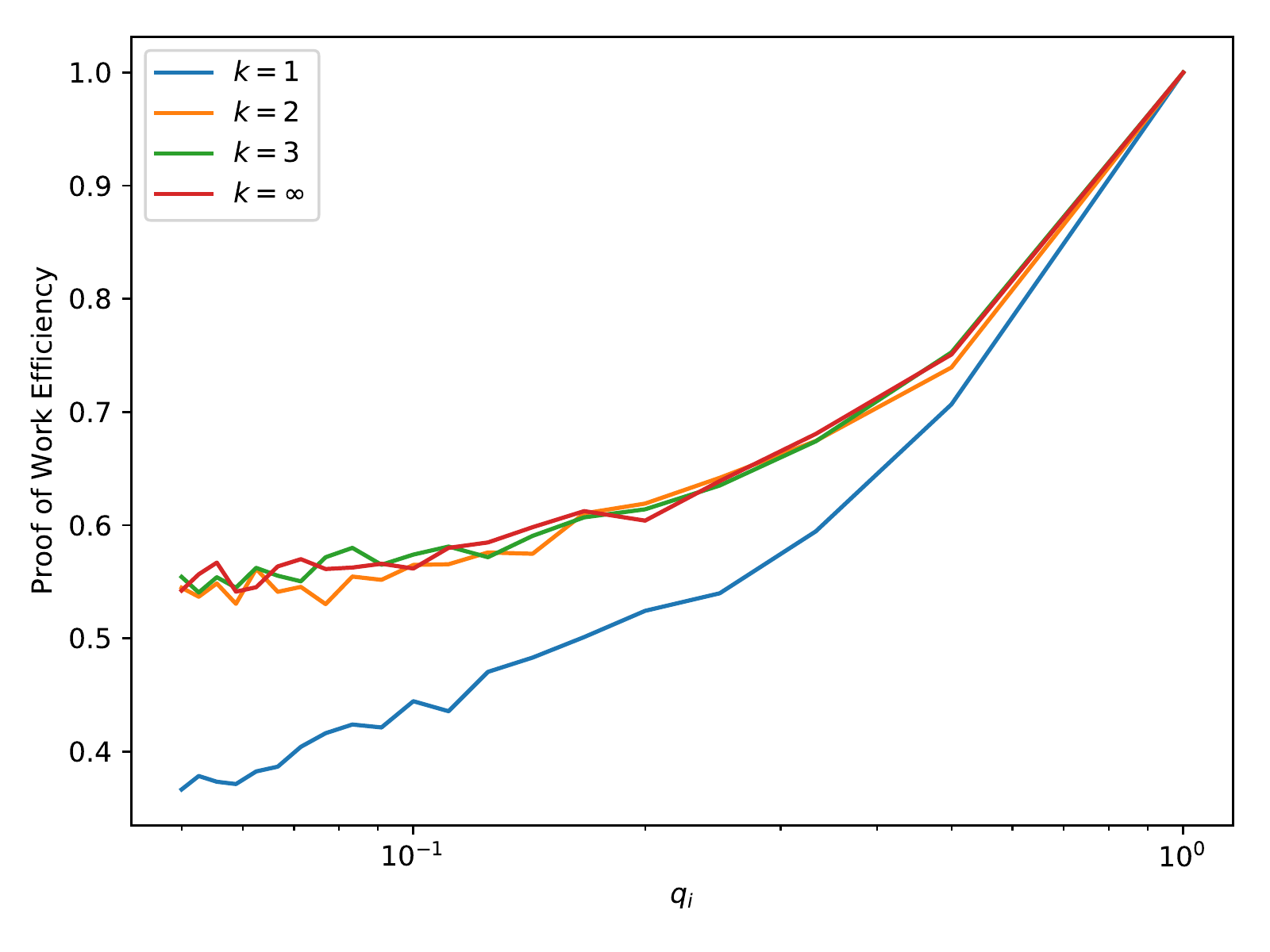}
    \includegraphics[width=.32\linewidth]{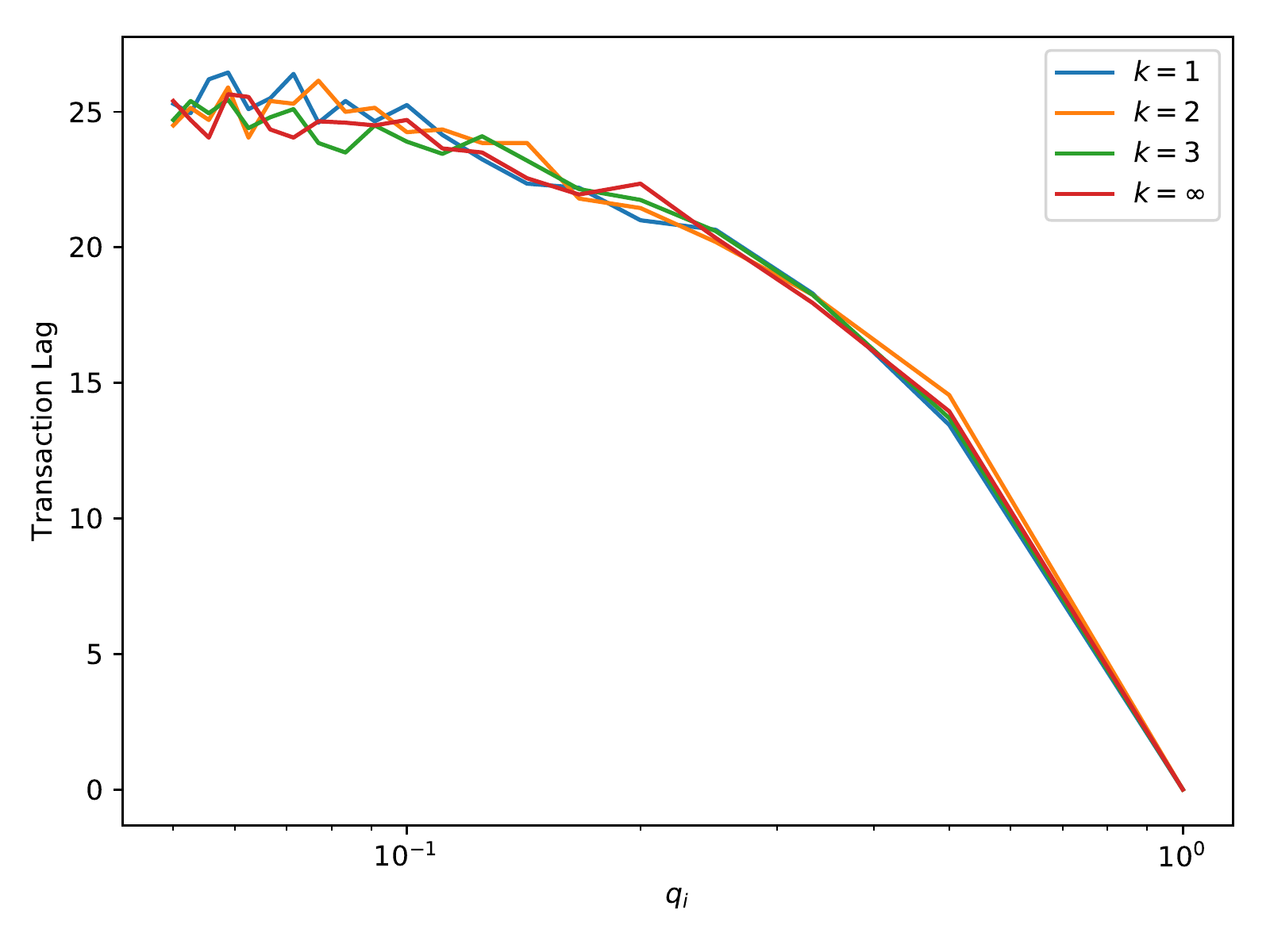}
    \caption{Performance Metrics for $n=4$ miners and $k \in \{1,2,3,\infty\}$}
    \label{fig:medres_variedN}
\end{figure}
First of all, we notice that for all parameter settings of $\mathcal{P}_{f,k}$, there exist information regimes where if each $q_i$ is low enough, the ledger suffers in its efficiency--even in the case where $k = \infty$. We also observe that increasing $k$ improves all metrics except lag, but not dramatically. For reasonable values of $q_i$, before fairness becomes an issue, there is a significant performance increase between $k=1$ and $k \ge 2$. However, $k>2$ is only really necessary for extremely small $q_i$. 

We also compared the performance for $n \in \{1, 2, \ldots, 20\}$ with $q = h = 1/n$, leading to similar results. Notably, as the number of miners grows the number of orphaned blocks decreases and the PoW efficiency improves with $k$.

\begin{figure}[h]
    \centering
    \includegraphics[width=.32\linewidth]{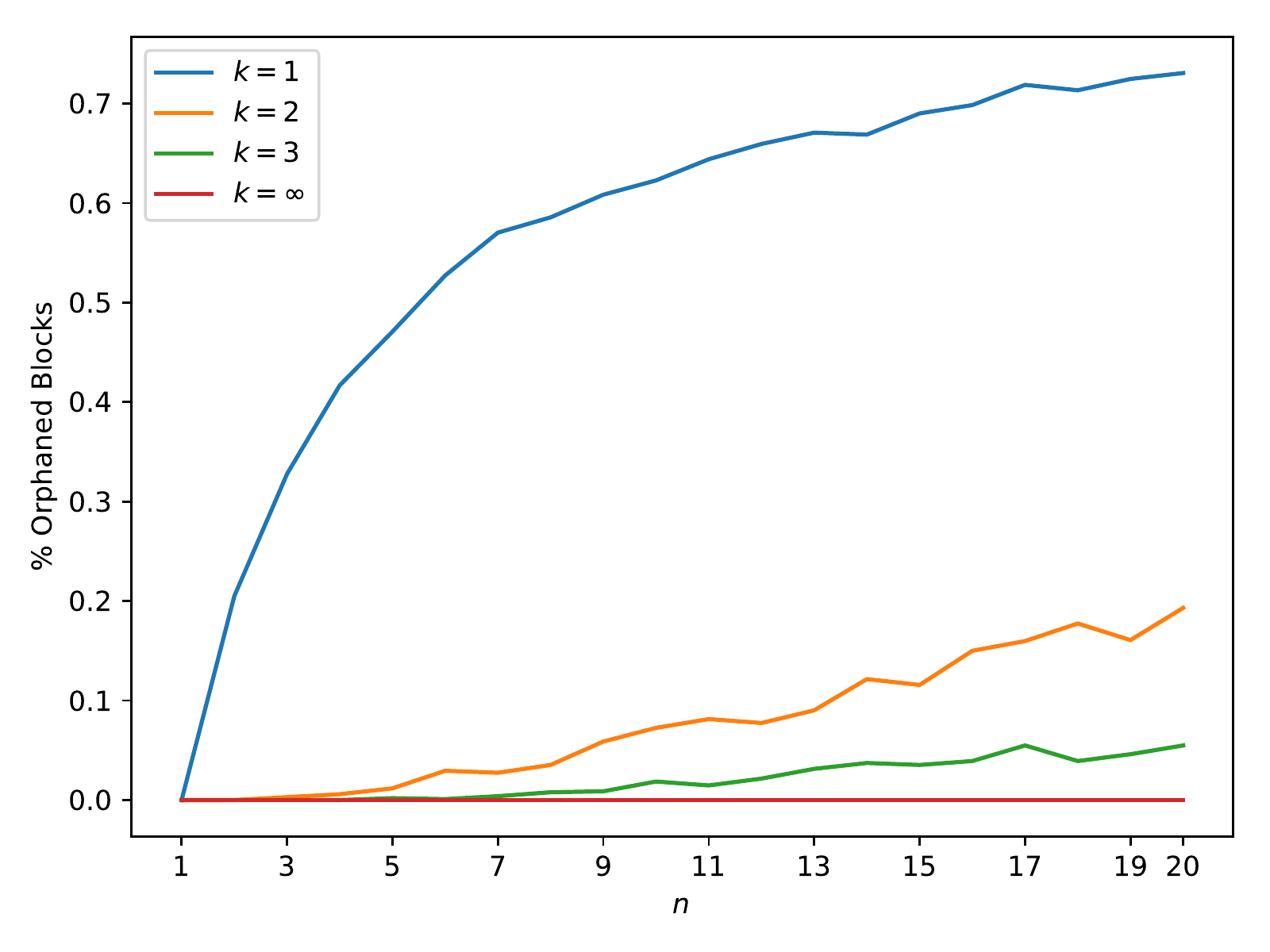}
    \includegraphics[width=.32\linewidth]{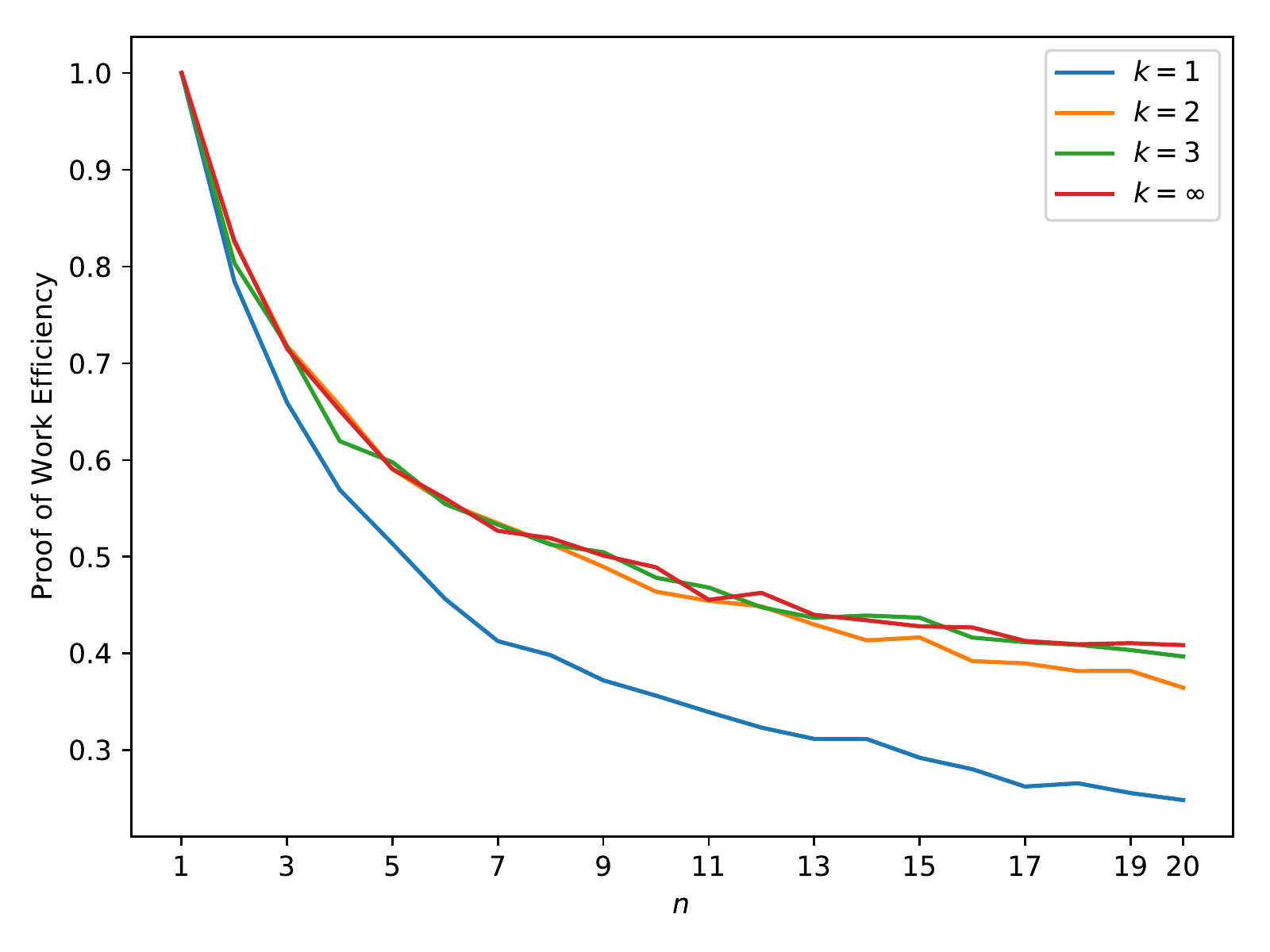}
    \caption{Performance Metrics for $n \in \{1, 2, \ldots, 20\}$ and $k \in \{1,2,3,\infty\}$}
    \label{fig:medres_variedN}
\end{figure}

\subsection{Dynamically adjusting $k$ and $f$} A key feature of Bitcoin is its dynamically adjusted difficulty. Our results suggest that a DAG-based ledger may also be able to dynamically adjust its internal parameters $k$ and $f$ to cope with changing transaction loads from end users. In fact, we see that high values of $k$ such as in SPECTRE do not provide much more of an added benefit to truncating the number of pointers at a smaller $k$. However, a dynamically adjusted protocol could sacrifice block size to make room for more pointers if efficiency is suffering in a period of high transaction loads to the ledger.

\section{Discussion / Future Directions}

Our results suggest inherent structural limitations to DAG-based ledgers. It remains to be seen whether the limitations formulated in this paper occur in practice. Given the generality of our DAG-based growth model, it would be interesting to study strategic considerations of $\mathcal{P}_{f,k}$ ledgers in their full generality, or augment the ledger space with a more complicated class of score functions for example.

In addition, it is known that there are other important limitations to Bitcoin. For example, the fact that every agent needs to keep a full copy of the ledger to make sure that validity is safe, or that PoW protocols result in an excessive use of energy resources around the globe. It would be interesting if the general approach of this paper could be applied to proposed solutions to these issues whereby one demonstrates inherent structural limitations in spite of all agents acting honestly.

\section{Acknowledgements}

Georgios Birmpas was supported by the European Research Council (ERC) under the advanced grant number 321171 (ALGAME) and the grant number 639945 (ACCORD). Philip Lazos was supported by the European Research Council (ERC) under the advanced grant number 321171 (ALGAME) and the advanced grant number 788893 (AMDROMA). Francisco J. Marmolejo-Coss\'io was supported by the Mexican National Council of Science and Technology (CONACyT). The authors would also like to thank Elizabeth Smith for the fruitful discussions during the preparation of this work.

\bibliographystyle{plainnat}
\bibliography{./refs.bib}

\end{document}